\documentclass[12pt,a4paper,onecolumn,draftcls]{IEEEtran}
\usepackage{epsfig,graphicx,subfigure,psfrag,amsmath,cases}
\usepackage{latexsym,amssymb,amsmath,epsfig,subfigure,algorithm,mathtools,amsthm}
\usepackage{algorithmic}
\usepackage{color}
\usepackage{url}
\usepackage{scrtime}
\usepackage{cite}
\usepackage{epstopdf}
\usepackage{hyperref}
\usepackage{subfigure}
\usepackage{bbding}
\usepackage{multicol}

\author{{Zhiqiang Wei, Derrick Wing Kwan Ng, Jinhong Yuan, and Hui-Ming Wang}
\thanks{Zhiqiang Wei, Derrick Wing Kwan Ng, and Jinhong Yuan are with the School of Electrical
Engineering and Telecommunications, the University of New South Wales, Australia (email: zhiqiang.wei@student.unsw.edu.au; w.k.ng@unsw.edu.au; j.yuan@unsw.edu.au).
Hui-Ming Wang is with the Ministry of Education Key Lab for Intelligent
Networks and Network Security, Xi¡¯an Jiaotong University, China (email: xjbswhm@gmail.com). This paper has been presented in part at the IEEE Globecom 2016 \cite{Wei2016NOMA}.\vspace*{-12mm}}}

\title{\vspace*{-10mm}Optimal Resource Allocation for Power-Efficient MC-NOMA with Imperfect Channel State Information}

\newtheorem{Thm}{Theorem}

\newtheorem{Def}{Definition}

\newtheorem{T-Prob}{Transformed Problem}

\DeclareMathOperator{\maxo}{maximize}
\DeclareMathOperator{\mino}{minimize}

\newtheorem{Remark}{Remark}

\newcommand{\abs}[1]{\lvert#1\rvert}

\begin{document}
\maketitle

\begin{abstract}
In this paper, we study power-efficient resource allocation for multicarrier non-orthogonal multiple access (MC-NOMA) systems.
The resource allocation algorithm design is formulated as a non-convex optimization problem which jointly designs the power allocation, rate allocation, user scheduling, and successive interference cancellation (SIC) decoding policy for minimizing the total transmit power.
The proposed framework takes into account the imperfection of channel state information at transmitter (CSIT) and quality of service (QoS) requirements of users.
To facilitate the design of optimal SIC decoding policy on each subcarrier, we define a \emph{channel-to-noise ratio outage threshold}.
Subsequently, the considered non-convex optimization problem is recast as a generalized linear multiplicative programming problem, for which a globally optimal solution is obtained via employing the branch-and-bound approach.
The optimal resource allocation policy serves as a system performance benchmark due to its high computational complexity.
To strike a balance between system performance and computational complexity, we propose a suboptimal iterative resource allocation algorithm based on difference of convex programming.
Simulation results demonstrate that the suboptimal scheme achieves a close-to-optimal performance. Also, both proposed schemes provide significant transmit power savings than that of conventional orthogonal multiple access (OMA) schemes.
\end{abstract}
%
%\renewcommand{\baselinestretch}{0.95}
%\normalsize

%\section{Motivations}
%\begin{itemize}
%  \item Under imperfect channel state information at transmitter side (CSIT), we aim to propose a downlink MC-NOMA scheme, which is robust against channel uncertainty and is optimal in the sense of minimizing total transmit power.
%  \item The resource allocation scheme should jointly optimize the power allocation, rate allocation, user scheduling, and successive interference cancellation (SIC) decoding order selection to minimize the total transmit power.
%\end{itemize}
\vspace*{-5mm}
\section{Introduction}
Recently, non-orthogonal multiple access (NOMA)
has received considerable attentions as a promising multiple access technique for the fifth-generation (5G) wireless communication networks\cite{Ding2015b,Dai2015,WeiSurvey2016}.
The basic rationale behind NOMA is to exploit the power domain for users multiplexing and to employ successive interference cancellation (SIC) at receivers to remove the multiple access interference (MAI).
In contrast to conventional orthogonal multiple access (OMA) schemes, NOMA is a promising solution to fulfil the demanding requirements of the 5G communication systems \cite{Andrews2014,wong2017key}, such as massive connectivity, low latency, high spectral efficiency, and enhanced user fairness.
In particular, NOMA can support overloaded transmission and increase the system throughput for given limited spectrum resources by enabling simultaneous transmission of multiple users utilizing the same frequency resources.
Also, multiple users with heterogeneous traffic requests can be served concurrently on the same frequency band to reduce the latency and to enhance the resource allocation fairness.
As a result, a preliminary version of NOMA, multiuser superposition transmission (MUST) scheme, has been proposed in the 3rd generation partnership project long-term evolution advanced (3GPP-LTE-A) networks\cite{Access2015}.

Meanwhile, in academia, extensive studies on NOMA schemes have been conducted.
In previous works \cite{Saito2013,Benjebbour2013,Xu2015,Ding2014,Dingtobepublished,Yang2016}, it has been shown that NOMA offers considerable performance gains over OMA in terms both of spectral efficiency and performance outage probability.
For instance, several pioneering works, e.g., \cite{Saito2013,Benjebbour2013}, have evaluated the performance of NOMA through system level simulations, showing that the overall cell throughput, cell-edge user throughput, and the degrees of fairness achieved by NOMA are all superior to those of OMA.
From an information theoretic perspective, the authors in \cite{Xu2015} proved that NOMA outperforms conventional time division multiple access (TDMA) schemes with a high probability in terms of both sum rate and individual rates.
In \cite{Ding2014}, asymptotic ergodic sum rate and outage probability for downlink NOMA systems with randomly deployed users were studied.
Furthermore, the impact of user pairing on performance of NOMA systems was characterized in \cite{Dingtobepublished}, where it has shown that the performance gain of NOMA over OMA in terms of sum rate can be enlarged substantially by pairing users with more distinctive channel conditions.
In \cite{Yang2016}, the authors analyzed the performance degradation of downlink NOMA systems in terms of outage probability and average sum rate due to partial channel state information at transmitter side (CSIT). Nevertheless, all the works on the performance analysis of NOMA schemes are based on the fixed power and subcarrier allocation, which is far from optimal resource allocation. Thus, the maximum potential performance gain brought by NOMA for 5G wireless networks remains unknown.

\textcolor[rgb]{0.00,0.00,0.00}{Resource allocation design plays a crucial role in exploiting the potential performance gain of NOMA systems, especially for multicarrier NOMA (MC-NOMA) systems\cite{Lei2016NOMA,Di2016sub,Sun2016Fullduplex}. Joint design of power and subcarrier allocation for MC-NOMA systems is generally NP-hard\cite{Lei2016NOMA}, and several works have made progresses on this to improve the sum rate employing
different optimization methods, such as the Lagrangian duality theory \cite{Lei2016NOMA}, matching game theory\cite{Di2016sub}, and monotonic optimization \cite{Sun2016Fullduplex}.
Apart from the maximization of sum rate, resource allocation with fairness considerations for NOMA systems have also been addressed in the existing works.
In \cite{Liu2015b}, the authors defined a scheduling factor and proposed a user pairing and power allocation scheme to achieve proportional fairness in resource allocation.
Given a predefined user group, the authors in \cite{Timotheou2015} studied the power allocation problem from a fairness standpoint by maximizing the minimum achievable user rate with instantaneous CSIT and minimizing the maximum outage probability by exploiting average CSIT.
Yet, the works in \cite{Lei2016NOMA,Di2016sub,Sun2016Fullduplex,Liu2015b} focused on resource allocation design of NOMA based on the assumption of perfect CSIT.
Unfortunately, it is unlikely to acquire the perfect CSIT due to channel estimation error, feedback delay, and quantization error.
For the case of NOMA with imperfect CSIT, it is difficult for a base station to sort the users' channel gains and to determine the user scheduling strategy and SIC decoding policy.
More importantly, the imperfect CSIT may cause a resource allocation mismatch, which may degrade the system performance.
Therefore, it is interesting and more practical to design robust resource allocation strategy for MC-NOMA systems taking into account of CSIT imperfectness.}

\textcolor[rgb]{0.00,0.00,0.00}{
In the literature, there are three commonly adopted methods to address the CSIT imperfectness for resource allocation designs, including no-CSIT \cite{JorswieckNOCSI}, worst-case optimization \cite{WangWorstCase}, and stochastic approaches \cite{ZhangStochastic}.
The assumption of no-CSIT usually results in a trivial equal power allocation strategy without any preference in resource allocation \cite{JorswieckNOCSI}.
Besides, it is pessimistic to assume no-CSIT since some sorts of CSIT, e.g., imperfect channel estimates or statistical CSIT, can be easily obtained in practical systems exploiting handshaking signals.
The worst-case based methods guarantee the system performance for the maximal CSIT mismatch \cite{WangWorstCase}.
However, an exceedingly large amount of system resources are exploited for some worst cases that rarely happen.
In particular, for our considered problem, the worst-case method leads to a conservative resource allocation design, which may translate into a higher power consumption.
On the other hand, the stochastic methods aim at modeling the CSIT and/or the channel estimation error according to the long term statistic of the channel realizations \cite{ZhangStochastic}.
It is more meaningful than the no-CSIT method since the statistical CSIT is usually available based on the long term measurements in practical systems.
More importantly, the stochastic methods can guarantee the average system performance over the channel realizations with moderate system resources.
Therefore, in this paper, we employed the stochastic method to robustly design the resource allocation strategy for MC-NOMA systems under imperfect CSIT.}

\textcolor[rgb]{0.00,0.00,0.00}{Recently, green radio design has become an important focus in both academia and industry due to the growing demands of energy consumption and the arising environmental concerns around the world\cite{wu2016overview}.
Note that power-efficient resource allocation has been extensively studied in OMA systems\cite{DerrickLargeAntenna,DerrickFD2016}.
The authors in \cite{DerrickLargeAntenna,DerrickFD2016} studied the power-efficient resource allocation design for large number of base station antennas and full-duplex radio base stations, respectively, while the proposed algorithms are not applicable for our considered MC-NOMA systems.
To address the green radio design for NOMA systems, the authors in \cite{zhang2016energy} proposed an optimal power allocation strategy for a single-carrier NOMA system to maximize the energy efficiency, while a separate subcarrier assignment and power allocation scheme was proposed for MC-NOMA systems in \cite{FangEnergyEfficientNOMAJournal}.
To minimize the total power consumption, the authors in \cite{Lei2016NPM} designed a suboptimal ``relax-then-adjust" algorithm for MC-NOMA systems.
Nevertheless, the existing designs in \cite{zhang2016energy,FangEnergyEfficientNOMAJournal,Lei2016NPM} are based on the assumption of perfect CSIT which may not be applicable to MC-NOMA systems under imperfect CSIT.
In our previous work \cite{Wei2016NOMA}, under statistical CSIT, we proposed an optimal SIC decoding policy for a two-user MC-NOMA system\footnote{\textcolor[rgb]{0.00,0.00,0.00}{In a two-user MC-NOMA system, there are at most two users multiplexing on each subcarrier.\vspace*{-10mm}}} and a suboptimal power-efficient resource allocation scheme with an equal rate assignment.
Furthermore, the authors in \cite{CuiPowerEfficientNOMA} proposed an optimal SIC decoding policy for single-carrier NOMA systems with more than two users multiplexing under statistical CSIT.
However, these two preliminary works \cite{Wei2016NOMA,CuiPowerEfficientNOMA} did not consider the joint resource allocation design for MC-NOMA systems and directly applying the designs to the considered MC-NOMA systems may lead to unsatisfactory performance.
%On the other hand, although it has been shown that NOMA possesses higher spectrum efficiency than that of OMA with limited power in the literatures\cite{Benjebbour2013,Ding2015b,Dai2015}, the power efficiency of NOMA for a given limited spectrum resource and quality of service (QoS) requirements were rarely discussed.
%In addition, the optimal SIC decoding policy for imperfect CSIT is still an unknown.
%More importantly, the imperfect CSIT may cause a resource allocation mismatch, which may degrade the power efficiency of the system.
To the best of the authors' knowledge, joint design of power allocation, rate allocation, user scheduling, and SIC decoding policy for power-efficient MC-NOMA under imperfect CSIT has not been reported yet.}

Based on aforementioned observations, in this paper, we study the power-efficient resource allocation design for downlink MC-NOMA systems under imperfect CSIT, where each user imposes its own QoS requirement.
The joint design of power allocation, rate allocation, user scheduling, and SIC decoding policy is formulated as a non-convex optimization problem to minimize the total transmit power.
To facilitate the design of optimal SIC decoding order, we define the \emph{channel-to-noise ratio (CNR) outage threshold}, which includes the joint effect of channel conditions and QoS requirements of users.
Based on the optimal SIC decoding policy, we propose an optimal resource allocation algorithm via the branch-and-bound (B\&B) approach \cite{Konno2000,horst2013global,MARANASProofBB,Androulakis1995}, which serves as a performance benchmark for MC-NOMA systems.
Furthermore, to strike a balance between system performance and computational complexity, we propose a suboptimal iterative resource allocation algorithm based on difference of convex (D.C.) programming\cite{dinh2010local,VucicProofDC}, which has a polynomial time computational complexity and converges quickly to a close-to-optimal solution.
Our simulation results show that the proposed resource allocation schemes enable significant transmit power savings and are robust against channel uncertainty.

%The rest of the paper is organized as follows. In Section II, we present the system model and formulate the joint resource allocation design as an optimization problem. In Section III, we proposed the optimal SIC decoding policy and reformulate the considered problem. Section IV and Section V present the optimal and suboptimal solutions, respectively. Simulation
%results are presented and analyzed in Section VI. Finally,
%Section VII concludes this paper.

Notations used in this paper are as follows. Boldface capital and lower case letters are reserved for matrices and vectors, respectively, ${\left( \cdot \right)^T}$ denotes the transpose of a vector or matrix.
$\mathbb{C}^{M\times N}$ denotes the set of all $M\times N$ matrices with complex entries; $\mathbb{R}^{M\times N}$ denotes the set of all $M\times N$ matrices with real entries;
$\abs{\cdot}$ denotes the absolute value of a complex scalar; $\left\| \cdot \right\|_2$ denotes the $l_2$-norm of a vector; \textcolor[rgb]{0.00,0.00,0.00}{$\left\lceil \cdot \right\rceil $ denotes the smallest integer greater than or equal to a real scalar; }and
$\Pr \left\{  \cdot  \right\}$ denotes the probability of a random event.
The circularly symmetric complex Gaussian distribution with mean $\mu$ and variance $\sigma^2$ is denoted by ${\cal CN}(\mu,\sigma^2)$;
$\sim$ stands for ``distributed as"; $U[a,b]$ denotes the uniform distribution in the interval $[a, b]$; and
${\nabla _{\mathbf{x}}}f$ denotes the gradient of a function $f$ with respective to (w.r.t.) vector ${\mathbf{x}}$.
\vspace*{-4mm}
\section{System Model and Problem Formulation}
In this section, after introducing the adopted MC-NOMA system model under imperfect CSIT, we define the QoS requirement based on outage probability and formulate the power-efficient resource allocation design as a non-convex optimization problem.
\vspace*{-4mm}
\subsection{System Model}
A downlink MC-NOMA system\footnote{\textcolor[rgb]{0.00,0.00,0.00}{In this paper, we focus on the power domain NOMA \cite{WeiSurvey2016} for the considered downlink communication scenario. Although the code-domain NOMA, such as sparse code multiple access (SCMA) \cite{Nikopour2013,Wang2015}, may outperform power-domain NOMA, SCMA is more suitable for the uplink communication where the reception complexity for information decoding is more affordable for base stations.}} with one base station (BS) and $M$ downlink users is considered and shown in Figure \ref{NOMA_model}.
All transceivers are equipped with single-antennas and there are $N_\mathrm{F}$ orthogonal subcarriers serving the $M$ users.
\textcolor[rgb]{0.00,0.00,0.00}{An overloaded scenario\footnote{\textcolor[rgb]{0.00,0.00,0.00}{Note that the proposed scheme in this paper can also be applied to underloaded systems where the number of subcarrier $N_\mathrm{F}$ is larger than the number of users $M$, i.e., $N_\mathrm{F} > M$.
For the sake of presentation, we first focus on the overloaded scenario and then apply the proposed resource allocation algorithm to both the overloaded and underloaded systems in the simulations.
Then, our simulation results in Section VI demonstrate that the proposed scheme is more power-efficient than that of the OMA scheme for both overloaded and underloaded systems.}\vspace*{-10mm}} is considered in this paper, i.e., $N_\mathrm{F}\le M$.
In addition, we assume that each of the $N_\mathrm{F}$ subcarriers can be allocated to at most two users via NOMA to reduce the computational complexity and delay incurred at receivers due to SIC decoding\footnote{\textcolor[rgb]{0.00,0.00,0.00}{In this paper, we focus on the two-user MC-NOMA system since it is more practical and is more appealing in both industry\cite{Access2015} and academia \cite{Dingtobepublished,LiuSWIPT,ChenTwoUser,Sun2016Fullduplex}.
The generalization of the proposed algorithms to the case of serving multiple users on each subcarrier is left for future work.}\vspace*{-10mm}}.
As a result, we have an implicit condition $\left\lceil {\frac{M}{2}} \right\rceil \le N_\mathrm{F} \le M$ such that the system can serve at least $M$ users.}
An example of a downlink MC-NOMA system with two users multiplexed on subcarrier $i$, $i \in \left\{ {1, \ldots ,N_\mathrm{F}} \right\}$, is illustrated in Figure \ref{NOMA_model}.
A binary indicator variable $s_{i,m} \in \{0,1\}$, $i \in \left\{ {1, \ldots ,N_\mathrm{F}} \right\}$, and $m \in \left\{ {1, \ldots ,M} \right\}$, is introduced as the user scheduling variable, where it is one if subcarrier $i$ is assigned to user $m$, and is zero otherwise.
Thus, we have the following constraint for $s_{i,m}$:
\vspace*{-2mm}
\begin{equation}\label{UserSchedulingConstraint}
\sum\limits_{m = 1}^M {{s_{i,m}}}  \le 2,\;\;\forall i.
\vspace*{-2mm}
\end{equation}

\begin{figure}[t]
\centering\vspace*{-8mm}
\includegraphics[width=3.5in]{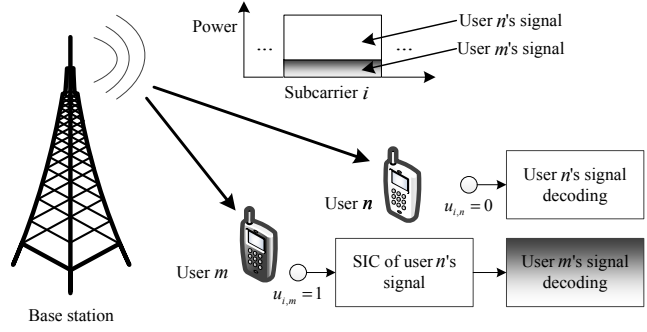}\vspace*{-7mm}
\caption{A downlink MC-NOMA system where user $m$ and user $n$ are multiplexed on subcarrier $i$. The base station transmits two superimposed signals with different powers. User $m$ is selected to perform SIC, i.e., $u_{i,m} = 1$, while user $n$ is not selected to, i.e., $u_{i,n} = 0$. User $m$ first decodes and removes the signal of user $n$ before decoding its desired signal, while user $n$ directly decodes its own signal with user $m$'s signal treated as noise.}\label{NOMA_model}\vspace*{-10mm}
\end{figure}

\noindent At the BS side, the transmitted signal on subcarrier $i$ is given by
\vspace*{-2mm}
\begin{equation}\label{Tx}
{x_i} = \sum\limits_{m = 1}^M {{s_{i,m}}\sqrt {{p_{i,m}}} {a_{i,m}}}, \;\; \forall i,
\vspace*{-2mm}
\end{equation}
where $a_{i,m}\in \mathbb{C}$ denotes the modulated symbol for user $m$ on subcarrier $i$ and $p_{i,m}$ is the allocated power for user $m$ on subcarrier $i$.
Different from most of the existing works on resource allocation of NOMA with perfect CSIT, e.g., \cite{sun2016optimal,Lei2015}, our model assumes imperfect CSIT. Under this condition, the BS needs to decide both the SIC decoding order, $u_{i,m} \in \left\{0,1\right\}$, and the rate allocation, $R_{i,m} > 0$, for each user on each subcarrier, which have critical impacts on the system power consumption.
The SIC decoding order variable is defined as follows:
\vspace*{-2mm}
\begin{equation}
u_{i,m} =
\left\{
\begin{array}{ll}
1 & \text{if}\;\text{user}\; m \;\text{on\;subcarrier}\; i \;\text{is selected to perform SIC,} \\
0 & \text{otherwise}.
\end{array}
\right.\label{SICVariable}\vspace*{-2mm}
\end{equation}

At the receiver side, the received signal at user $m$ on subcarrier $i$ is given by
\vspace*{-2mm}
\begin{equation}\label{Rx}
{y_{i,m}} = {h_{i,m}}\sum\limits_{n = 1}^M {{s_{i,n}}\sqrt {{p_{i,n}}} {a_{i,n}}}  + {z_{i,m}},
\vspace*{-2mm}
\end{equation}
where $z_{i,m} \in \mathbb{C}$ denotes the additive white Gaussian noise (AWGN) for user $m$ on subcarrier $i$ with a zero-mean and variance $\sigma^2_{i,m}$, i.e., $z_{i,m} \sim {\cal CN}(0,\sigma^2_{i,m})$.
\textcolor[rgb]{0.00,0.00,0.00}{Variable ${h}_{i,m} = \frac{{g}_{i,m}}{\sqrt{\text{PL}_m}} \in \mathbb{C}$ denotes the channel coefficient between the BS and user $m$ on subcarrier $i$ capturing the joint effect of path loss and small scale fading.
In particular, $\text{PL}_m$ denotes the path loss of user $m$, and we assume that the BS can accurately estimate the path loss of each user $\text{PL}_m$, $\forall m$, based on the long term measurements.
On the other hand, ${g}_{i,m} \sim {\cal CN}(0,1) $ denotes the small scale fading, which is modeled as Rayleigh fading in this paper \cite{Tse2005}.
Due to the channel estimation error and/or feedback delay, only imperfect CSIT is available for resource allocation.
To capture the channel estimation error, we model the channel coefficient for user $m$ on subcarrier $i$ as
\vspace*{-3.6mm}
\begin{equation}\label{ChannelModel}
h_{i,m} = \hat{h}_{i,m} + \Delta h_{i,m},
\vspace*{-3.6mm}
\end{equation}
where $\hat{h}_{i,m}$ denotes the estimated channel coefficient for user $m$ on subcarrier $i$, $\Delta h_{i,m} \sim {\cal CN}(0,\frac{\kappa^2_{i,m}}{\text{PL}_m})$ denotes the corresponding CSIT error, and $\frac{\kappa^2_{i,m}}{\text{PL}_m} >0 $ denotes the variance of the channel estimation error. We assume that the channel estimates $\hat{h}_{i,m}$ and the channel estimation error $\Delta h_{i,m}$ are uncorrelated.}
According to the SIC decoding order policy, both multiplexed users will choose to perform SIC or directly decode its own messages.
%As illustrated in Figure \ref{NOMA_model}, according to the SIC decoding order policy, one of the paired users on each subcarrier will perform SIC, and the other user will directly decode its own desired signal.

%\textcolor[rgb]{1.00,0.00,0.00}{We should state more about rat allocation and SIC decoding order.
%If we need the discrete rate allocation, we should state here. Label: 2016-08-18.}

%Therefore, the posterior distribution of channel coefficient based on the estimated CSIT can be given by
%\begin{equation}
%h_{i,m} \sim {\cal CN}(\hat{h}_{i,m},\kappa^2_{i,m}).
%\end{equation}
\vspace*{-4mm}
\subsection{QoS Requirements}
\vspace*{-1mm}
To facilitate our design, we define an outage probability on each subcarrier, which is commonly adopted in the literature for resource allocation design \cite{Zhu2009,Kwan_AF_2010}.
We assume that if the SIC of any user is failed, the user cannot decode its own messages, and thus an outage event occurs\cite{Ding2014}.
Therefore, if user $m$ on subcarrier $i$ is selected to perform SIC, i.e., $u_{i,m} = 1$, we have the outage probability as follows:
\vspace*{-1.5mm}
\begin{equation}\label{OutageProbabilityWithSIC}
\hspace*{-2mm}{\mathrm{P}}_{i,m}\hspace*{-1mm}=\hspace*{-1mm} {\mathrm{Pr}}\hspace*{-1mm}\left\{\hspace*{-1mm} {C_{i,m}^{{\mathrm{SIC}}} \hspace*{-1mm}<\hspace*{-5mm} \sum\limits_{n = 1,n \ne m}^M \hspace*{-5mm}{{s_{i,n}}R_{i,n}}} \hspace*{-1mm}\left| {{{\hat h}_{i,m}}}, u_{i,m} \hspace*{-1mm}=\hspace*{-1mm} 1 \right.\hspace*{-1mm}\right\}
\hspace*{-0.8mm}+\hspace*{-0.8mm} {\mathrm{ Pr}}\hspace*{-1mm}\left\{\hspace*{-1mm} {C_{i,m}^{{\mathrm{SIC}}} \hspace*{-1mm}\ge\hspace*{-5mm} \sum\limits_{n = 1,n \ne m}^M \hspace*{-5mm}{{s_{i,n}}R_{i,n}},{C_{i,m}^{(1)}} \hspace*{-1mm}<\hspace*{-1mm} {R_{i,m}}} \hspace*{-1mm}\left| {{{\hat h}_{i,m}}}, u_{i,m} \hspace*{-1mm}=\hspace*{-1mm} 1\right.\hspace*{-1mm}\right\}\hspace*{-1mm},
\vspace*{-1.5mm}
\end{equation}
where ${\mathrm{P}}_{i,m}$ denotes the outage probability of user $m$ on subcarrier $i$ due to the channel uncertainty and $R_{i,m}$ denotes the allocated rate of user $m$ on subcarrier $i$.
Variable $C_{i,m}^{{\mathrm{SIC}}}$ denotes the achievable rate of user $m$ for decoding the interference from the other user on subcarrier $i$ and $C_{i,m}^{(1)}$ denotes the achievable rate of user $m$ on subcarrier $i$ for decoding its own message with a successful SIC. \textcolor[rgb]{0.00,0.00,0.00}{In our model, with two users multiplexing on subcarrier $i$, $C_{i,m}^{{\mathrm{SIC}}}$ and $C_{i,m}^{(1)}$ are given by
\vspace*{-3mm}
\begin{equation}
C_{i,m}^{\mathrm{SIC}} = {\log _2}\left( {1 + \frac{{{{\left| {{h_{i,m}}} \right|}^2}\sum\limits_{n = 1,\;n \ne m}^M {{s_{i,n}}{p_{i,n}}} }}{{{s_{i,m}}{p_{i,m}}{{\left| {{h_{i,m}}} \right|}^2} + \sigma _{i,m}^2}}} \right)\; \text{and} \;
C_{i,m}^{(1)} = {\log _2}\left( {1 + \frac{{{s_{i,m}}{p_{i,m}}{{\left| {{h_{i,m}}} \right|}^2}}}{{\sigma _{i,m}^2}}} \right),\label{CapacityWithSIC}
\vspace*{-2.5mm}
\end{equation}
respectively. Note that there is only one non-zero entry in the summations in \eqref{OutageProbabilityWithSIC} and \eqref{CapacityWithSIC}, which means that the interference to be cancelled during the SIC processing arises from only one user due to the constraint in \eqref{UserSchedulingConstraint}.}

On the other hand, if user $m$ on subcarrier $i$ is not selected to perform SIC, i.e., $u_{i,m} = 0$, we have the outage probability as follows:
\vspace*{-4mm}
\begin{equation}\label{OutageProbabilityWithoutSIC}
{\mathrm{P}}_{i,m}' = {\rm{Pr}}\left\{ {C_{i,m}^{(2)} < {R_{i,m}}\left| {{{\hat h}_{i,m}}}, u_{i,m} = 0\right.} \right\},
\vspace*{-4mm}
\end{equation}
where ${\mathrm{P}}_{i,m}'$ denotes the outage probability of user $m$ on subcarrier $i$ due to the channel uncertainty. Variable ${C_{i,m}^{(2)}}$ denotes the achievable rate of user $m$ on subcarrier $i$ for decoding its own messages without performing SIC and it is given by
\vspace*{-3mm}
\begin{equation}\label{CapacityWithoutSIC}
{C_{i,m}^{(2)}} = {\log _2}\left( {1 + \frac{{{s_{i,m}}{p_{i,m}}{{\left| {{h_{i,m}}} \right|}^2}}}{{{{\left| {{h_{i,m}}} \right|}^2}\sum\limits_{n=1,\;n \ne m}^M {{s_{i,n}}{p_{i,n}}}  + \sigma _{i,m}^2}}} \right).
\vspace*{-2mm}
\end{equation}

Now, we define the QoS requirement of user $m$ on subcarrier $i$ as follows:
\vspace*{-4mm}
\begin{equation}\label{QoSConstraint}
{\mathrm{P}}_{i,m}^{{\mathrm{out}}} = s_{i,m} \left\{u_{i,m}{\mathrm{P}}_{i,m} + \left( {1 - {u_{i,m}}} \right){\mathrm{P}}_{i,m}' \right\} \le {\delta _{i,m}},\;\; \forall i,m,
\vspace*{-3mm}
\end{equation}
where $\delta_{i,m}$, $0 \le \delta_{i,m} \le 1$, denotes the required outage probability of user $m$ on subcarrier $i$.
When user $m$ is not assigned on subcarrier $i$, i.e., $s_{i,m} = 0$, this inequality is always satisfied.

In OMA systems, subcarrier $i$ will be allocated to user $m$ exclusively, i.e., $\sum\limits_{m = 1}^M {{s_{i,m}}}  = 1$.
In this case, we have ${C_{i,m}^{(1)}} = {C_{i,m}^{(2)}}$, and ${\mathrm{P}}_{i,m}'$  denotes the corresponding outage probability of user $m$ on subcarrier $i$.
Also, performing SIC at user $m$ is not required, and thus we have $u_{i,m}=0$,  ${\mathrm{P}}_{i,m}^{{\mathrm{out}}} =  {\mathrm{P}}_{i,m}'$.
In other words, the outage probability defined in the considered MC-NOMA system generalizes that of OMA systems as a subcase.

\vspace*{-4mm}
\subsection{Optimization Problem Formulation}
We aim to jointly design the power allocation, rate allocation, user scheduling, and SIC
decoding policy for minimizing the total transmit power of the considered downlink MC-NOMA system under imperfect CSIT.
The joint resource allocation design can be formulated as the following optimization problem:
\vspace*{-2mm}
\begin{align} \label{P1}
%\left( {{\mathrm{P}_0}}\right)
&\underset{\mathbf{s},\;\mathbf{u},\;\mathbf{p},\;\mathbf{r}}{\mino}\,\, \,\, \notag \sum\limits_{m = 1}^M {\sum\limits_{i = 1}^{N_\mathrm{F}} {s_{i,m}p_{i,m}}}\\[-2mm]
\notag\mbox{s.t.}\;\;
%%%%%
&\mbox{\textbf{C1}: } s_{i,m} \in \left\{0,\;1\right\},\;\forall i,m,
\;\;\;\;\;\;\;\;\;\mbox{\textbf{C2}: } u_{i,m} \in \left\{0,\;1\right\},\;\forall i,m,
\;\;\;\;\;\;\;\;\;\mbox{\textbf{C3}: } p_{i,m} \ge 0,\;\forall i,m,\;\;\notag\\[-3mm]
&\mbox{\textbf{C4}: } R_{i,m} \ge 0,\;\forall i,m,
\hspace*{18.9mm}\mbox{\textbf{C5}: } {\mathrm{P}}_{i,m}^{{\mathrm{out}}} \le {\delta _{i,m}},\; \forall i,m,\notag\\[-2mm]
&\mbox{\textbf{C6}: } \sum\limits_{i = 1}^{N_{\mathrm{F}}} {{s_{i,m}}{{R}_{i,m}}}  \ge R_m^{\mathrm{total}},\;\forall m,
\;\hspace*{-0.46mm}\mbox{\textbf{C7}: } \sum\limits_{m = 1}^M {{s_{i,m}}}  \le 2,\;\forall i,
\end{align}
\par
\vspace{-2mm}
\noindent where $\mathbf{s} \in \mathbb{R}^{N_{\mathrm{F}}M \times 1}$, $\mathbf{u} \in \mathbb{R}^{N_{\mathrm{F}} M\times1}$, $\mathbf{p} \in \mathbb{R}^{N_{\mathrm{F}}M \times1}$, and $\mathbf{r} \in \mathbb{R}^{N_{\mathrm{F}} M \times 1}$ denote the sets of optimization variables $s_{i,m}$, $u_{i,m}$, $p_{i,m}$, and $R_{i,m}$.
Constraints \textbf{C1} and \textbf{C2} restrict the user scheduling variables and SIC decoding order variables to be binary, respectively.
Constraints \textbf{C3} and \textbf{C4} ensure the non-negativity of the power allocation variables and the rate allocation variables, respectively.
Constraint \textbf{C5} is the QoS constraint of outage probability for user $m$ on subcarrier $i$. According to \eqref{QoSConstraint}, \textbf{C5} is inactive when $s_{i,m}=0$.
%As for $s_{i,m}=1$, C5 limits the outage probability with SIC when $u_{i,m}=1$ and restricts the outage probability without SIC when $u_{i,m}=0$.
\textcolor[rgb]{0.00,0.00,0.00}{In \textbf{C6}, constant $R_m^{\mathrm{total}} > 0$ denotes the required minimum total data rate of user $m$.
In particular, constraint \textbf{C6} is introduced for rate allocation such that the required minimum total data rate of user $m$ can be guaranteed.
We note that the rate allocation $R_{i,m}$ for user $m$ on subcarrier $i$ may be very low, but the achievable rate of user $m$ is bounded below by $R_m^{\mathrm{total}}$.}
Constraint \textbf{C7} is imposed to ensure that at most two users are multiplexed on each subcarrier.
Note that our problem includes OMA as a subcase when $\sum\limits_{m = 1}^M {{s_{i,m}}} = 1$ in \textbf{C7}.

This problem in \eqref{P1} is a mixed combinatorial non-convex optimization problem.
In general, there is no systematic and computational efficient approach to solve \eqref{P1} optimally.
The combinatorial nature comes from the binary constraints \textbf{C1} and \textbf{C2} while the non-convexity arises in the QoS constraint in \textbf{C5}.
Besides, the coupling between binary variables and continuous variables in constraint \textbf{C5} yields an intractable problem.
However, we note that the power allocation variables and SIC decoding order variables are only involved in the QoS constraint \textbf{C5} among all the constraints.
%Given any user scheduling and rate allocation strategy in the feasible solution set, there might be some underlying relationships we can obtain from the QoS constraint C5.
Therefore, by exploiting this property, we attempt to simplify the optimization problem in \eqref{P1} to facilitate the resource allocation design in the next section.
\vspace*{-4mm}
\section{Problem Transformation}
In this section, we first propose the optimal SIC policy on each subcarrier taking into account the impact of imperfect CSIT.
Then, the minimum total transmit power per subcarrier is derived.
Subsequently, we transform the optimization problem which paves the way for the design of the optimal resource allocation.
\subsection{Optimal SIC Policy Per Subcarrier}
%In OMA scenario, there is no need to perform SIC. In NOMA scenario, given user $m$ and user $n$ multiplexing on the single subcarrier $i$, i.e., $s_{i,m}=s_{i,n}=1$, we intend to derive the optimal SIC decoding order to simplify the problem by eliminating the optimization variable $u_{i,m}$ and $u_{i,n}$.

Under perfect CSIT, for a two-user NOMA downlink system, it is well-known that the optimal SIC decoding order is the descending order of channel gains for maximizing the total system sum rate \cite{Dingtobepublished}.
However, under imperfect CSIT, it is not possible to decide the SIC decoding order by comparing the actual channel gains between the multiplexed users.
To facilitate the design of resource allocation for the case of NOMA with imperfect CSIT, we define an \emph{channel-to-noise ratio (CNR) outage threshold} in the following, from which we can decide the optimal SIC decoding order to minimize the total transmit power.

\begin{Def}[CNR Outage Threshold]
For user $m$ on subcarrier $i$ with the estimated channel coefficient $\hat{h}_{i,m}$, the noise power ${{\sigma _{i,m}^2}}$, and the required outage probability ${\delta _{i,m}}$, a CNR outage event occurs when the CNR, $\frac{{{{\left| {{h_{i,m}}} \right|}^2}}}{{\sigma _{i,m}^2}}$, is smaller than a CNR outage threshold ${\beta _{i,m}}$. The CNR outage probability of user $m$ on subcarrier $i$ can be written as follows:
\vspace*{-2mm}
\begin{equation}\label{OutageThreshold1}
{\mathrm{Pr}}\left\{ {\frac{{{{\left| {{h_{i,m}}} \right|}^2}}}{{\sigma _{i,m}^2}} < {\beta _{i,m}}} \left| {{{\hat h}_{i,m}}} \right. \right\} = {\delta _{i,m}},\;\;\forall i,m.
\vspace*{-2mm}
\end{equation}
Therefore, the CNR outage threshold is given by
\vspace*{-1mm}
\begin{equation}\label{OutageThreshold2}
{\beta _{i,m}} = \frac{{F_{\left. {{{\left| {{h_{i,m}}} \right|}^2}} \right|{{\hat h}_{i,m}}}^{ - 1}\left( {{\delta _{i,m}}} \right)}}{{\sigma _{i,m}^2}},\;\;\forall i,m,
\vspace*{-2mm}
\end{equation}
where ${F_{\left. {{{\left| {{h_{i,m}}} \right|}^2}} \right|{{\hat h}_{i,m}}}}\left( x \right)$, $x \ge 0$, is the conditional cumulative distribution function\footnote{We note that if the fading channel is not Rayleigh distributed, the proposed schemes in this paper are still applicable whereas we only need to update ${F_{\left. {{{\left| {{h_{i,m}}} \right|}^2}} \right|{{\hat h}_{i,m}}}}\left( x \right)$ according to the distribution of small scale fading.} (CDF) of channel power gain of user $m$ on subcarrier $i$.
In fact, according to the channel model in \eqref{ChannelModel}, ${F^{-1}_{\left. {{{\left| {{h_{i,m}}} \right|}^2}} \right|{{\hat h}_{i,m}}}}\left( x \right)$ is the inverse of a noncentral chi-square CDF\footnote{The inverse function of a noncentral chi-square CDF can be computed efficiently by standard numerical solvers or implemented as a look-up table for implementation.\vspace*{-10mm}} with degrees of freedom of 2 and a noncentrality parameter of \vspace*{-1mm} $\frac{{\left| {{\hat h}_{i,m}} \right|^2}}{\kappa^2_{i,m}}{\text{PL}_m}$, ${\kappa^2_{i,m}}>0$. Note that if perfect CSIT is available, i.e., ${\kappa^2_{i,m}}=0$, there is no CNR outage caused by channel estimation error and we have $\beta_{i,m} = \frac{{{{\left| {{\hat h_{i,m}}} \right|}^2}}}{{\sigma _{i,m}^2}} = \frac{{{{\left| {{ h_{i,m}}} \right|}^2}}}{{\sigma _{i,m}^2}}$.
\end{Def}

We note that the CNR outage defined in \eqref{OutageThreshold1} is different from that of the outage events defined in \eqref{OutageProbabilityWithSIC} and \eqref{OutageProbabilityWithoutSIC}.
In particular, the former does not involve the required target data rate of users, while the latter outage event occurs when the achievable rate is smaller than a given target data rate.
In the rest of the paper, we refer to CNR outage as in \eqref{OutageThreshold1} if it is stated explicitly.
Otherwise, it refers to the outage defined as in \eqref{OutageProbabilityWithSIC} and \eqref{OutageProbabilityWithoutSIC}.
The CNR outage threshold defined in \eqref{OutageThreshold2} involves the estimated channel gain, the channel estimation error distribution, the noise power, and the required outage probability.
It captures the joint effect of channel conditions and QoS requirements in the statistical sense for imperfect CSIT scenarios.
Specifically, the user with higher CNR outage threshold may have better channel condition and/or lower required outage probability, and vice versa.
In fact, the CNR outage threshold serves as a criterion for the optimal SIC decoding order, which is summarized in the following theorem.

\begin{Thm}[Optimal SIC Decoding Order]\label{Theorem1}
Given user $m$ and user $n$ multiplexing on subcarrier $i$,
%the optimal SIC decoding order for the minimization of the total transmit power is the descending order of CNR outage thresholds.
%In other words, if ${{\beta _{i,m}}} \geq {{\beta _{i,n}}}$, selecting user $m$ to perform SIC will consume a smaller total transmit power compared to selecting user $n$. Therefore,
the optimal SIC decoding order is determined by the CNR outage thresholds as follows:
\vspace*{-6mm}
\begin{equation}
\left(u_{i,m},u_{i,n}\right) =
\left\{
\begin{array}{ll}
\left(1,0\right) & \text{if}\;\beta _{i,m} \ge \beta _{i,n}, \\
\left(0,1\right) & \text{if}\;\beta _{i,m} < \beta _{i,n}.
\end{array}
\right.\label{SICDecodingOrderNOMA}
\vspace*{-1.5mm}
\end{equation}
which means that the user with a higher CNR outage threshold will perform SIC decoding.
\end{Thm}
  \emph{Proof: }
Please refer to Appendix A for a proof of Theorem \ref{Theorem1}.\qed

%\vspace*{-2mm}
%\begin{equation}
%u_{i,m} =
%\left\{
%\begin{array}{ll}
%1 & \text{if}\;\beta _{i,m} \ge \beta _{i,n}, \\
%0 & \text{if}\;\beta _{i,m} < \beta _{i,n}, \\
%\end{array}
%\right.\; \text{and}\;
%u_{i,n} =
%\left\{
%\begin{array}{ll}
%1 & \text{if}\; \beta _{i,n} \ge \beta _{i,m}, \\
%0 & \text{if}\; \beta _{i,n} < \beta _{i,m}. \\
%\end{array}
%\right.\label{SICDecodingOrderNOMA}
%\vspace*{-5mm}
%\end{equation}

\textcolor[rgb]{0.00,0.00,0.00}{With only a single user allocated on subcarrier $i$, i.e., $\sum\limits_{n = 1}^M {{s_{i,n}}} = 1$, it reduces to the OMA scenario and there is no need of SIC decoding for all the users, and thus we have
\vspace*{-2.5mm}
\begin{equation}\label{SICDecodingOrderOMA1}
u_{i,m} = 0, \;\forall m,\;\;\text{if}\;\sum\limits_{n = 1}^M {{s_{i,n}}} = 1.
\vspace*{-1mm}
\end{equation}
Note that the optimal SIC decoding policies defined in \eqref{SICDecodingOrderNOMA} and \eqref{SICDecodingOrderOMA1} are conditioned on any given feasible user scheduling strategy satisfying constraints \textbf{C1}, \textbf{C6}, and \textbf{C7} in (11) for minimizing the total transmit power.}
%We note that the optimal SIC decoding order in Theorem \ref{Theorem1} is proposed to minimize the total transmit power on an arbitrary single subcarrier for the problem in \eqref{P1}.
More importantly, given any point in the feasible solution set spanned by \textbf{C1}, \textbf{C4}, \textbf{C6}, and \textbf{C7} in \eqref{P1}, our proposed optimal SIC decoding order always consumes the minimum total transmit power to satisfy the QoS constraint \textbf{C5}.
\textcolor[rgb]{0.00,0.00,0.00}{By exploiting constraint \textbf{C5}, the underlying relationship between SIC decoding order variables $u_{i,m}$ and user scheduling variables $s_{i,m}$ is revealed for NOMA and OMA scenarios in \eqref{SICDecodingOrderNOMA} and \eqref{SICDecodingOrderOMA1}, respectively.}

\begin{Remark}\label{Remark1}
In Theorem \ref{Theorem1}, it is noteworthy that the optimal SIC decoding order only depends on the CNR outage threshold, $\beta _{i,m}$, and is independent of the target data rates of users.
This observation is reasonable. Let us first recall from the basic principle for SIC decoding for the case of NOMA with perfect CSIT\cite{Tse2005} and then extend it to the case of imperfect CSIT.
Specifically, for perfect CSIT, the strong user (with a higher channel gain) can decode the messages of the weak user (with a lower channel gain), if we can guarantee that the weak user can decode its own messages, no matter who has a higher target data rate.
This is due to the fact that the achievable rate for the strong user to decode the messages of the weak user is always higher than that of the weak user to decode its own.
On the other hand, if the weak user performs SIC, due to its worse channel condition, a higher transmit power is required such that the strong user's messages are decodable at the weak user, no matter whose target data rate is higher.
Similarly, given user $m$ and user $n$ multiplexed on subcarrier $i$ under imperfect CSIT, according to \eqref{OutageProbabilityWithSIC}, \eqref{OutageProbabilityWithoutSIC}, and \eqref{OutageThreshold1}, we have the following implication under the condition of $\beta _{i,m} \hspace*{-1mm} \ge\hspace*{-1mm}  \beta _{i,n}$:
\vspace*{-2.5mm}
\begin{equation}
{\rm{Pr}}\left\{\hspace*{-0.5mm} {C_{i,n}^{(2)} \hspace*{-1mm}<\hspace*{-1mm} {R_{i,n}}\hspace*{-1mm}\left| {{{\hat h}_{i,n}}}, u_{i,n} \hspace*{-1mm}=\hspace*{-1mm} 0\right.} \hspace*{-1mm}\right\} \hspace*{-1mm}\le\hspace*{-1mm} {\delta _{i,n}} \Rightarrow
{\mathrm{Pr}}\left\{\hspace*{-0.5mm} C_{i,m}^{{\mathrm{SIC}}} \hspace*{-1mm}<\hspace*{-1mm} {R_{i,n}} \hspace*{-1mm}\left| {{{\hat h}_{i,m}}}, u_{i,m} \hspace*{-1mm}=\hspace*{-1mm} 1 \right.\hspace*{-1mm}\right\} \hspace*{-1mm}\le\hspace*{-1mm} {\delta _{i,m}}, \; \text{if}\;\beta _{i,m} \hspace*{-1mm}\ge\hspace*{-1mm} \beta _{i,n},
\vspace*{-2mm}
\end{equation}
which means that the SIC process at the user with a higher CNR outage threshold will always satisfy its QoS constraint if the other user's (with a lower CNR outage threshold) decoding process satisfy its own QoS constraint, no matter whose target data rate is higher.
Also, if the user with a lower CNR outage threshold performs SIC, it requires an extra power to guarantee the QoS constraint of the SIC process, no matter who requires a higher data rate.
Therefore, the optimal SIC decoding order is determined according to the CNR outage thresholds, and it is independent of the target data rates.
\end{Remark}
%Therefore, the SIC decoding order variables $u_{i,m}$ can be eliminated from ${{\mathrm{P}_0}}$.
\vspace*{-4mm}
\subsection{Minimum Total Transmit Power Per Subcarrier}
In this section, we exploit Theorem \ref{Theorem1} to further simplify the problem in \eqref{P1} via expressing the power allocation variables $p_{i,m}$ in terms of CNR outage threshold and target data rate.
Given user $m$ and user $n$ multiplexed on subcarrier $i$, i.e., $s_{i,m}=s_{i,n}=1$, according to the proof of Theorem \ref{Theorem1} in Appendix A, the minimum total transmit power required on subcarrier $i$ to satisfy the QoS constraint \textbf{C5} in \eqref{P1} can be generally represented as:
\vspace*{-3mm}
\begin{equation}\label{TotalPowerConsumptionNOMA}
p^{\mathrm{total}}_{(i,m,n)} = \frac{{{\gamma _{i,m}}}}{{{\beta _{i,m}}}} + \frac{{{\gamma _{i,n}}}}{{{\beta _{i,n}}}} + \frac{{{\gamma _{i,m}}{\gamma _{i,n}}}}{{\max \left( {{\beta _{i,m}},{\beta _{i,n}}} \right)}},
\vspace*{-2mm}
\end{equation}
where the subscript $({i,m,n})$ denotes that users $m$ and $n$ are multiplexed on subcarrier $i$, $\gamma _{i,m}$ denotes the required signal-to-interference-plus-noise ratio (SINR) to support the target data rate of user $m$ on subcarrier $i$, and it is given by ${\gamma _{i,m}} = {2^{{R_{i,m}}}} - 1$. Particularly, the power allocations for users $m$ and $n$ are given according to their CNR outage threshold as follows:
%\vspace*{-2mm}
%\begin{equation}
%p_{i,m} =
%\left\{
%\begin{array}{ll}
%\frac{{{\gamma _{i,m}}}}{{{\beta _{i,m}}}} & \text{if}\; \beta _{i,m} \ge \beta _{i,n}, \\
%\frac{{{\gamma _{i,m}}}}{{{\beta _{i,m}}}} + \frac{{{\gamma _{i,m}}{\gamma _{i,n}}}}{\beta _{i,n}} & \text{if}\; \beta _{i,m} < \beta _{i,n}, \\
%\end{array}
%\right. \; \text{and} \;\;
%p_{i,n} =
%\left\{
%\begin{array}{ll}
%\frac{{{\gamma _{i,n}}}}{{{\beta _{i,n}}}} & \text{if}\; \beta _{i,n} \ge \beta _{i,m}, \\
%\frac{{{\gamma _{i,n}}}}{{{\beta _{i,n}}}} + \frac{{{\gamma _{i,n}}{\gamma _{i,m}}}}{\beta _{i,m}} & \text{if}\; \beta _{i,n} < \beta _{i,m}. \\
%\end{array}
%\right.\label{PowerAllocation4}
%\vspace*{-2mm}
%\end{equation}
\vspace*{-1mm}
\begin{equation}\label{PowerAllocation4}
\left(p_{i,m},\;p_{i,n}\right) =
\left\{
\begin{array}{ll}
\left(\frac{{{\gamma _{i,m}}}}{{{\beta _{i,m}}}},\; \frac{{{\gamma _{i,n}}}}{{{\beta _{i,n}}}} + \frac{{{\gamma _{i,n}}{\gamma _{i,m}}}}{\beta _{i,m}}\right) & \text{if}\; \beta _{i,m} \ge \beta _{i,n}, \\
\left(\frac{{{\gamma _{i,m}}}}{{{\beta _{i,m}}}} + \frac{{{\gamma _{i,m}}{\gamma _{i,n}}}}{\beta _{i,n}},\; \frac{{{\gamma _{i,n}}}}{{{\beta _{i,n}}}}\right) & \text{if}\; \beta _{i,m} < \beta _{i,n},
\end{array}
\right.
\vspace*{-2mm}
\end{equation}

We note that, in the OMA scenario, if $\sum\limits_{n = 1}^M {{s_{i,n}}} = 1$ and $s_{i,m}=1$, the required minimum transmit power to satisfy the QoS constraint \textbf{C5} in \eqref{P1} is given by
\vspace*{-3mm}
\begin{equation}\label{TotalPowerConsumptionOMA}
p^{\mathrm{total}}_{(i,m)} =p_{i,m}= \frac{{{\gamma _{i,m}}}}{{{\beta _{i,m}}}},
\vspace*{-2mm}
\end{equation}
where the subscript $({i,m})$ denotes that users $m$ is assigned on subcarrier $i$ exclusively.
%Besides, the total transmit power of NOMA and OMA both depend on the outage threshold and target rate.
By combining the user scheduling variables of subcarrier $i$, $s_{i,m},\;\forall m \in \left\{ {1, \ldots ,M} \right\}$, with \eqref{TotalPowerConsumptionNOMA}, the total transmit power on subcarrier $i$ can be generally expressed as:
\vspace*{-2mm}
\textcolor[rgb]{0.00,0.00,0.00}{\begin{equation}\label{TotalPowerConsumptionUniform}
p^{\mathrm{total}}_{i} = \sum\limits_{m = 1}^M {\frac{{{s_{i,m}}{\gamma _{i,m}}}}{{{\beta _{i,m}}}}}  + \sum\limits_{m = 1}^{M-1} {\sum\limits_{n = m + 1}^M {\frac{{{s_{i,m}}{\gamma _{i,m}}{s_{i,n}}{\gamma _{i,n}}}}{{\max \left( {{\beta _{i,m}},{\beta _{i,n}}} \right)}}} },
\end{equation}}
\par
\vspace*{-7mm}
\noindent
where \eqref{TotalPowerConsumptionUniform} subsumes the cases of the power consumption in the OMA scenario in \eqref{TotalPowerConsumptionOMA}.
In fact, equation \eqref{TotalPowerConsumptionUniform} unveils the relationship between the required minimum total transmit power on subcarrier $i$, the user scheduling variables, and the allocated data rates.
Since the SIC decoding process on each subcarrier is independent with each other, we have the total transmit power of all the subcarriers as follows:
\vspace*{-2.1mm}
\begin{equation}\label{TotalPowerConsumptionUniform2}
p^{\mathrm{total}} = \sum\limits_{i = 1}^{{N_{\mathrm{F}}}}{p^{\mathrm{total}}_{i}} = \sum\limits_{i = 1}^{{N_{\mathrm{F}}}} {\sum\limits_{m = 1}^M {\frac{{{s_{i,m}}{\gamma _{i,m}}}}{{{\beta _{i,m}}}}} }  + \sum\limits_{i = 1}^{{N_{\mathrm{F}}}} {\sum\limits_{m = 1}^{M-1} {\sum\limits_{n = m + 1}^M {\frac{{{s_{i,m}}{\gamma _{i,m}}{s_{i,n}}{\gamma _{i,n}}}}{{\max \left( {{\beta _{i,m}},{\beta _{i,n}}} \right)}}} } }.
\vspace*{-2mm}
\end{equation}

Recall that Theorem \ref{Theorem1} is derived from the QoS constraint \textbf{C5} in \eqref{P1}.
Therefore, given any user scheduling and rate allocation strategy in the feasible solution set spanned by constraints \textbf{C1}, \textbf{C4}, \textbf{C6}, and \textbf{C7}, we obtain the optimal SIC decoding order and power allocation solution from \eqref{SICDecodingOrderNOMA}, \eqref{SICDecodingOrderOMA1}, \eqref{PowerAllocation4}, and \eqref{TotalPowerConsumptionOMA}, which can satisfy the QoS constraint \textbf{C5} with the minimum power consumption in \eqref{TotalPowerConsumptionUniform2}.
In other words, the problem in \eqref{P1} can be transformed equivalently to a simpler one to minimize the power consumption in \eqref{TotalPowerConsumptionUniform2} w.r.t. the user scheduling and rate allocation variables.

%Therefore, power allocation variables $p_{i,m}$ can also be eliminated from ${{\mathrm{P}_0}}$.

\begin{Remark}\label{Remark1}
Comparing the minimum total transmit power for NOMA and OMA in \eqref{TotalPowerConsumptionNOMA} and \eqref{TotalPowerConsumptionOMA}, respectively, we obtain $p^{\mathrm{total}}_{(i,m,n)} > p^{\mathrm{total}}_{(i,m)} + p^{\mathrm{total}}_{(i,n)}$.
At the first sight, it seems that OMA is more power-efficient than NOMA as the third term in \eqref{TotalPowerConsumptionNOMA} is the extra power cost for users' multiplexing.
However, this comparison is unfair for NOMA since both schemes require different spectral efficiencies.
To keep the same spectral efficiency in the considered scenarios, the required target data rate for OMA users in \eqref{TotalPowerConsumptionOMA} should be doubled\cite{Ding2015b,sun2016optimal}.
\textcolor[rgb]{0.00,0.00,0.00}{Due to the combinatorial nature of user scheduling, it is difficult to prove that the proposed MC-NOMA scheme is always more power-efficient than the OMA scheme with the optimal resource allocation strategy.
%In fact, in our previous work \cite{Wei2016NOMA}, we have proved that NOMA is more power-efficient than OMA in a single-carrier scenario under statistical CSIT.
%Following the method in \cite{Wei2016NOMA}, it can be easily proved that the proposed scheme has a better power saving performance than OMA in a single-carrier scenario under imperfect CSIT.
In fact, for a special case with $M = 2N_{\mathrm{F}}$, this conclusion can be proved mathematically based on our previous work in \cite{Wei2016NOMA}.
For the general case, we rely on the simulation results to demonstrate the power savings of our proposed schemes compared to the OMA scheme.}
\end{Remark}
\vspace*{-4mm}
\subsection{Problem Transformation}
Based on Theorem \ref{Theorem1}, the problem in \eqref{P1} is equivalent to the following optimization problem:
\vspace*{-8.5mm}
\begin{align}\label{P2}
&\underset{\mathbf{s},\;\boldsymbol{\gamma}}{\mino}\,\, \,\, \notag \sum\limits_{i = 1}^{{N_{\mathrm{F}}}} {\sum\limits_{m = 1}^M {\frac{{{s_{i,m}}{\gamma _{i,m}}}}{{{\beta _{i,m}}}}} }  + \sum\limits_{i = 1}^{{N_{\mathrm{F}}}} {\sum\limits_{m = 1}^{M-1} {\sum\limits_{n = m + 1}^M {\frac{{{s_{i,m}}{\gamma _{i,m}}{s_{i,n}}{\gamma _{i,n}}}}{{\max \left( {{\beta _{i,m}},{\beta _{i,n}}} \right)}}} } } \\[-4mm]
\mbox{s.t.}\;\;
%%%%%
&\mbox{\textbf{C1}},\;\mbox{\textbf{C7}},\notag\\[-5mm]
&\mbox{\textbf{C4}: } \gamma_{i,m} \ge 0,\;\;\forall i,m,
\;\;\mbox{\textbf{C6}: } \sum\limits_{i = 1}^{{N_{\mathrm{F}}}} {s_{i,m}}{{{\log }_2}\left( {1 +{\gamma _{i,m}}} \right)}  \ge {{R}_m^{\mathrm{total}}},\;\;\forall m,
\end{align}
\par
\vspace{-2mm}
\noindent where the rate allocation variables ${{R}_{i,m}}$ are replaced by their equivalent optimization variables $\gamma_{i,m}$ in \textbf{C4} and \textbf{C6}, and $\boldsymbol{\gamma}  \in \mathbb{R}^{N_{\mathrm{F}} M \times 1} $ denotes the set of $\gamma_{i,m}$.

\textcolor[rgb]{0.00,0.00,0.00}{The reformulated problem in \eqref{P2} is simpler than that of the problem in \eqref{P1}, since the number of optimization variables is reduced and the QoS constraint \textbf{C5} is safely removed. However, \eqref{P2} is also difficult to solve.}
In particular, \textbf{C1} are binary constraints, and there are couplings between the binary variables and continuous variables in both the objective function and constraint \textbf{C6}.
In fact, the problem in \eqref{P2} is a non-convex mixed-integer nonlinear programming problem (MINLP), which is NP-hard\cite{floudas1995nonlinear} in general. Now, we transform the problem from \eqref{P2} to:
\vspace*{-2mm}
\begin{align}\label{P3}
&\underset{\mathbf{{s}},\;\boldsymbol{{\gamma}}}{\mino}\,\, \,\, \notag \sum\limits_{i = 1}^{{N_{\mathrm{F}}}} {\sum\limits_{m = 1}^M {\frac{{{{\gamma}}_{i,m}}}{{{\beta _{i,m}}}}} }  + \sum\limits_{i = 1}^{{N_{\mathrm{F}}}} {\sum\limits_{m = 1}^{M-1} {\sum\limits_{n = m + 1}^M {\frac{{{{\gamma}}_{i,m}{{\gamma}}_{i,n}}}{{\max \left( {{\beta _{i,m}},{\beta _{i,n}}} \right)}}} } } \\[-3mm]
\notag\mbox{s.t.}\;\;
%%%%%
&\mbox{\textbf{C1}, \textbf{C4}, \textbf{C7},} \notag\\[-3mm]
&\widetilde{\text{\textbf{C6}}}\mbox{: } \sum\limits_{i = 1}^{{N_{\mathrm{F}}}} {{{\log }_2}\left( {1 +{{\gamma} _{i,m}}} \right)}  \ge {{R}_m^{\mathrm{total}}},\;\;\forall m,
\;\;\mbox{\textbf{C8}: } {{\gamma}}_{i,m} = {{s}_{i,m}} {\gamma}_{i,m},\;\;\forall i,m.
\end{align}
\par
\vspace*{-2mm}
\noindent
Constraint \textbf{C8} is imposed to preserve the original couplings between the binary variables and the continuous variables. Constraint $\widetilde{\text{\textbf{C6}}}$ is obtained from constraint \textbf{C6} in \eqref{P2} by employing the equality ${s_{i,m}}{{{\log }_2}\left( {1 +{\gamma _{i,m}}} \right)} = {{{\log }_2}\left( {1 +{s_{i,m}}{{\gamma} _{i,m}}} \right)}$ for ${s}_{i,m} \in \left\{0,\;1\right\}$. Clearly, the problem in \eqref{P3} is equivalent to the problem in \eqref{P2}, when we substitute constraint \textbf{C8} into the objective function and constraint $\widetilde{\text{\textbf{C6}}}$. Therefore, in the sequel, we focus on solving \eqref{P3}.

%In the following, based on the transformed problem in \eqref{P3}, an optimal algorithm is proposed as a performance benchmark via employing the branch and bound (B\&B) approach\cite{Konno2000,horst2013global,MARANASProofBB}.
%Then, to strike a balance between computational complexity and system performance, in Section \ref{Suboptimal}, a suboptimal resource allocation algorithm is also proposed via exploiting the difference of convex structure of the reformulated problem in \eqref{P2}.
\vspace*{-4mm}
\section{Optimal Solution}
In this section, to facilitate the design of the optimal resource allocation algorithm, we first relax the binary constraint \textbf{C1} and augment the coupling constraint \textbf{C8} into objective function by introducing a penalty factor.
Then, an optimal resource allocation algorithm is proposed based on B\&B approach\cite{Konno2000,horst2013global,MARANASProofBB}.
\vspace*{-4mm}
\subsection{Continuous Relaxation and Penalty Method}
To start with, we relax the binary constraint on $s_{i,m}$ in \textbf{C1} which yields:
\vspace*{-2mm}
\begin{align}\label{P3Continuous} &\underset{\mathbf{\overline{s}},\;\boldsymbol{{\gamma}}}{\mino}\,\, \,\, \notag \sum\limits_{i = 1}^{{N_{\mathrm{F}}}} {\sum\limits_{m = 1}^M {\frac{{{{\gamma}}_{i,m}}}{{{\beta _{i,m}}}}} }  + \sum\limits_{i = 1}^{{N_{\mathrm{F}}}} {\sum\limits_{m = 1}^{M-1} {\sum\limits_{n = m + 1}^M {\frac{{{{\gamma}}_{i,m}{{\gamma}}_{i,n}}}{{\max \left( {{\beta _{i,m}},{\beta _{i,n}}} \right)}}} } } \\[-2mm]
\notag\mbox{s.t.}\;\;
%%%%%
& \mbox{\textbf{C4}},\;\widetilde{\text{\textbf{C6}}},\notag\\[-4mm]
&\overline{\text{\textbf{C1}}} \mbox{: } 0 \le \overline{s}_{i,m} \le 1,\;\;\forall i,m,
\;\;\overline{\text{\textbf{C7}}} \mbox{: } \sum\limits_{m = 1}^M {{\overline{s}_{i,m}}} \le 2,\;\;\forall i,
\;\;\overline{\text{\textbf{C8}}} \mbox{: } {{\gamma}}_{i,m} = {\overline{s}_{i,m}} {\gamma}_{i,m},\;\;\forall i,m,
\end{align}
\par
\vspace*{-2mm}
\noindent
where ${\overline{s}_{i,m}}$ denotes the continuous relaxation of the binary variable ${s}_{i,m}$ and $\mathbf{\overline{s}} \in \mathbb{R}^{N_{\mathrm{F}} M \times 1} $ denotes the set of ${\overline{s}_{i,m}}$. $\overline{\text{\textbf{C1}}}$, $\overline{\text{\textbf{C7}}}$, and $\overline{\text{\textbf{C8}}}$ denote the modified constraints for \textbf{C1}, \textbf{C7}, and \textbf{C8} via replacing ${s}_{i,m}$ with ${\overline{s}_{i,m}}$, correspondingly.
In general, the solution of the constraint relaxed problem in \eqref{P3Continuous} provides a lower bound for the problem in \eqref{P3}.
However, by utilizing the coupling relationship in constraint $\overline{\text{\textbf{C8}}}$, the following theorem states the equivalence between \eqref{P3} and \eqref{P3Continuous}.
%In particular, in $\overline{\text{C8}}$, ${\overline{s}_{i,m}}$ has to be one when ${{\gamma}}_{i,m}>0$, and can be any value between zero and one when ${{\gamma}}_{i,m}=0$.
%Also, the following theorem provides a way to recover the optimal solution of \eqref{P3} by solving \eqref{P3Continuous}.
\begin{Thm}\label{Theorem2}
The relaxed problem in \eqref{P3Continuous} is equivalent to the problem in \eqref{P3}.
More importantly, for the optimal solution of \eqref{P3Continuous}, $\left({{\overline{s}_{i,m}^*}},{\gamma}_{i,m}^*\right)$, $i \in \left\{1, \ldots ,{N_{\rm{F}}}\right\}$, $m \in \left\{ 1, \ldots ,M \right\}$, the optimal solution of \eqref{P3} can be recovered via keeping ${\gamma}_{i,m}^*$ and performing the following mapping:
\vspace*{-2mm}
\begin{equation}\label{OptimalConvert2}
{{{s}_{i,m}^*}} =
\left\{
\begin{array}{ll}
1 & \text{if}\;{\gamma}_{i,m}^* > 0,\\
0 & \text{if}\;{\gamma}_{i,m}^* = 0.
\end{array}
\right.
\vspace*{-1.5mm}
\end{equation}

\end{Thm}

\begin{proof}
Please refer to Appendix B for a proof of Theorem \ref{Theorem2}.
\end{proof}

Note that the equivalence between $\widetilde{\text{\textbf{C6}}}$ in \eqref{P3Continuous} and \textbf{C6} in \eqref{P2} still holds at the optimal solution via the mapping relationship in \eqref{OptimalConvert2}.
It is notable that the reformulation in \eqref{P3} not only transforms the couplings between binary variables and continuous variables into a single constraint \textbf{C8}, but also reveals the special structure that enables the equivalence between \eqref{P3} and \eqref{P3Continuous}.
Now, the non-convexity remaining in \eqref{P3Continuous} arises from the product term in both objective function and constraint $\overline{\text{\textbf{C8}}}$.
Thus, we augment $\overline{\text{\textbf{C8}}}$ into the objective function via introducing a penalty factor $\theta$ as follows:
\vspace*{-4mm}
\begin{equation} \label{P3Penalty}
\underset{\mathbf{\overline{s}},\;\boldsymbol{{\gamma}}}{\mino}\,\, \,\, G^{\theta}\left(\mathbf{\overline{s}},{\boldsymbol{{\gamma} }} \right) \;\;\;
\mbox{s.t.}\;\;
%%%%%
\overline{\text{\textbf{C1}}},\; \text{\textbf{C4}},\; \widetilde{\text{\textbf{C6}}},\; \overline{\text{\textbf{C7}}},
\vspace*{-2mm}
\end{equation}
where the new objective function is given by
\vspace*{-1mm}
\begin{equation}\label{NewObj}
G^{\theta}\left(\mathbf{\overline{s}},{\boldsymbol{{\gamma} }} \right) =
\sum\limits_{i = 1}^{{N_{\mathrm{F}}}} {\sum\limits_{m = 1}^M {\frac{{{{\gamma}}_{i,m}}}{{{\beta _{i,m}}}}} }  + \sum\limits_{i = 1}^{{N_{\mathrm{F}}}} {\sum\limits_{m = 1}^{M-1} {\sum\limits_{n = m + 1}^M {\frac{{{{\gamma}}_{i,m}{{\gamma}}_{i,n}}}{{\max \left( {{\beta _{i,m}},{\beta _{i,n}}} \right)}}} } } +\theta \sum\limits_{i = 1}^{{N_{\mathrm{F}}}} {\sum\limits_{m = 1}^M {\left({{\gamma}}_{i,m} - {\overline{s}_{i,m}} {\gamma}_{i,m}\right)} }.\vspace*{-1mm}
\end{equation}

\begin{Thm}\label{Theorem3}
If the problem in \eqref{P3Continuous} is feasible with a bounded optimal value, the problem in \eqref{P3Penalty} is equivalent to \eqref{P3Continuous} for a sufficient large penalty factor $\theta \gg 1$.
\end{Thm}

\begin{proof}
Please refer to \cite{DerrickFD2016,sun2016optimal} for a proof of Theorem \ref{Theorem3}.
\end{proof}

The problem in \eqref{P3Penalty} is a generalized linear multiplicative programming problem over a compact convex set, where the optimal solution can be obtained via the B\&B method\cite{Konno2000}.

\vspace*{-4mm}
\subsection{B\&B Based Optimal Resource Allocation Algorithm}
The B\&B method has been widely adopted as a partial enumeration strategy for global optimization\cite{horst2013global}.
The basic principle of B\&B relies on a successive subdivision of the original region (Branch) that systematically discards non-promising subregions via employing lower bound or upper bound (Bound).
It has been proved that B\&B can converge to a globally optimal solution in finite numbers of iterations if the branching operation is consistent and the selection operation is bound improving\footnote{We note that the B\&B method cannot be directly used on the problem in \eqref{P2} since a tight convex bounding function for the objective function has not been reported in the literatures and its feasible solution set is not compact. Based on Theorem \ref{Theorem2} and Theorem \ref{Theorem3}, we transform the problem in \eqref{P2} to a equivalent generalized linear multiplicative programming problem on a convex compact feasible solution set in \eqref{P3Penalty}, which can be handled by the B\&B method\cite{Konno2000}.\vspace*{-10mm}}\cite{horst2013global,MARANASProofBB}.
In this section, we first propose the branching rule and the bounding method for the problem in \eqref{P3Penalty}, and then develop the optimal resource allocation algorithm.

\subsubsection{Branching Procedure}
From constraint $\widetilde{\text{\textbf{C6}}}$ in \eqref{P3Penalty}, it can be observed that ${{\gamma}}_{i,m} > 2^{{R}_m^{\mathrm{total}}}-1$ is not the optimal rate allocation since the objective function is monotonically increasing with ${{\gamma}}_{i,m}$.
Therefore, we rewrite constraint \textbf{C4} in \eqref{P3Penalty} with a box constraint, \textbf{C4}: $0 \le {\gamma}_{i,m} \le 2^{{R}_m^{\mathrm{total}}}-1$, $\forall i,m$.
As a result, the optimization variables ${\mathbf{\overline{s}}}$ and $\boldsymbol{{\gamma}}$ are defined in a hyper-rectangle, which is spanned by $\overline{\text{\textbf{C1}}}$ and \textbf{C4}.
For notational simplicity, we redefine the optimization variables in \eqref{P3Penalty} as follows:
\vspace*{-6mm}
\begin{equation}
{v_{i,m}} = {{\gamma} _{i,m}}\; \text{and} \; {v_{i,m+M}} = {\overline{s}_{i,m}}, \;\;\forall{i} \in \left\{1, \ldots ,{N_{\mathrm{F}}}\right\},\;\forall{m} \in \left\{ 1, \ldots ,M \right\}. \label{VariableRedefine1}
\vspace*{-4mm}
\end{equation}
Then the product terms in $G^{\theta}\left(\mathbf{\overline{s}},{\boldsymbol{{\gamma} }} \right)$ in \eqref{P3Penalty}, i.e., ${{{\gamma}}_{i,m}{{\gamma}}_{i,n}}$ and $-{\overline{s}_{i,m}} {\gamma}_{i,m}$, can be generally represented by $a_{i,m,n}{{v_{i,m}v_{i,n}}}$, where $a_{i,m,n} \in \{1,-1\}$ is a constant coefficient.
%For example, we have $a_{i,m,n} = \left(2^{{R}_m^{\mathrm{total}}}-1\right)\left(2^{{R}_n^{\mathrm{total}}}-1\right)$ for ${{{\gamma}}_{i,m}{{\gamma}}_{i,n}}$, and $a_{i,m,m+M} = \left(2^{{R}_m^{\mathrm{total}}}-1\right)$ for ${\overline{s}_{i,m}} {\gamma}_{i,m}$.
With the definition in \eqref{VariableRedefine1}, we can use the new variable ${v_{i,m}}$ and the original variable ${\left({{\gamma}}_{i,m},\;{\overline{s}_{i,m}}\right)}$ interchangeably in the rest of the paper.
Furthermore, the hyper-rectangle spanned by constraints $\overline{\text{\textbf{C1}}}$ and \textbf{C4} can be presented by $\Phi = \left[ v _{i,m}^{{\mathrm{L}}},v _{i,m}^{{\mathrm{U}}} \right]$, $\forall{i} \in \left\{1, \ldots ,{N_{\mathrm{F}}}\right\}$, $\forall{m} \in \left\{ 1, \ldots ,2M \right\}$, where ${v _{i,m}^{\mathrm{L}}}$ and ${v_{i,m}^{\mathrm{U}}}$ denote the lower bound and upper bound for ${{v _{i,m}}}$, respectively.

\begin{figure}[t]
\centering\vspace*{-8mm}
\includegraphics[width=3.2in]{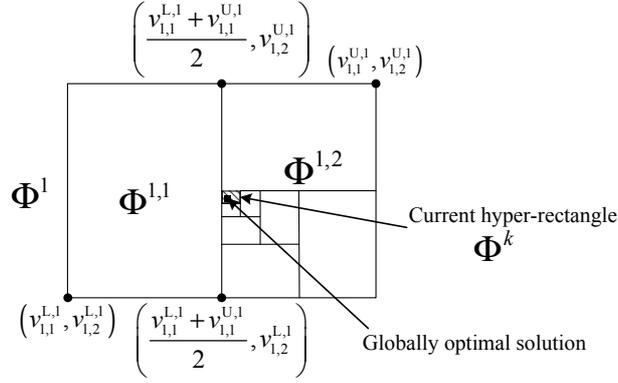}\vspace*{-8mm}
\caption{An illustration of the successive branching procedure in a two-dimensional space.}\vspace*{-9mm}
\label{BnB:a}
\end{figure}

An successive branching procedure with bisection on the longest edge of the hyper-rectangle is adopted in this paper\cite{Konno2000}, as illustrated in Figure \ref{BnB:a}.
Particularly, in the first iteration, according to constraints $\overline{\text{\textbf{C1}}}$ and \textbf{C4} in \eqref{P3Penalty}, the initial hyper-rectangle $\Phi^1$ is characterized by
\vspace*{-4mm}
\begin{equation}\label{InitialRectangle}
v _{i,m}^{{\mathrm{L}},1} = 0,\;
v _{i,m}^{{\mathrm{U}},1} = 2^{{R}_m^{\mathrm{total}}}-1,\;
v _{i,m+M}^{{\mathrm{L}},1} = 0,\;
\text{and}\;
v _{i,m+M}^{{\mathrm{U}},1} = 1.
\vspace*{-3mm}
\end{equation}
Then, in the $k$-th iteration, the current hyper-rectangle\footnote{The current hyper-rectangle selection rule will be presented in the overall algorithm, cf. \textbf{Algorithm} \ref{alg1}.\vspace*{-10mm}} $\Phi^k = \left[ v _{i,m}^{{\mathrm{L}},k},v _{i,m}^{{\mathrm{U}},k} \right]$, $\forall{i} \in \left\{1, \ldots ,{N_{\mathrm{F}}}\right\}$, $\forall{m} \in \left\{ 1, \ldots ,2M \right\}$, is partitioned into the following two subrectangles:
\vspace*{-2mm}
\begin{equation}\label{Subdivision}
\Phi^{k,1} = \left[ {\begin{array}{*{20}{c}}
{v_{1,1}^{{\mathrm{L}},k}}&{v_{1,1}^{{\mathrm{U}},k}}\vspace*{-1mm}\\
 \vdots & \vdots \vspace*{-1mm}\\
{v_{{i^k},{m^k}}^{{\mathrm{L}},k}}&{\frac{{v_{{i^k},{m^k}}^{{\mathrm{L}},k} + v_{{i^k},{m^k}}^{{\mathrm{U}},k}}}{2}}\vspace*{-1mm}\\
 \vdots & \vdots \vspace*{-1mm}\\
{v_{{N_{\mathrm{F}}},2M}^{{\mathrm{L}},k}}&{v_{{N_{\mathrm{F}}},2M}^{{\mathrm{U}},k}}
\end{array}} \right]\; \text{and} \;
\Phi^{k,2} = \left[ {\begin{array}{*{20}{c}}
{v_{1,1}^{{\mathrm{L}},k}}&{v_{1,1}^{{\mathrm{U}},k}}\vspace*{-1mm}\\
 \vdots & \vdots \vspace*{-1mm}\\
{\frac{{v_{{i^k},{m^k}}^{{\mathrm{L}},k} + v_{{i^k},{m^k}}^{{\mathrm{U}},k}}}{2}}&{v_{{i^k},{m^k}}^{{\mathrm{U}},k}}\vspace*{-1mm}\\
 \vdots & \vdots \vspace*{-1mm}\\
{v_{{N_{\mathrm{F}}},2M}^{{\mathrm{L}},k}}&{v_{{N_{\mathrm{F}}},2M}^{{\mathrm{U}},k}}
\end{array}} \right],
\vspace*{-2mm}
\end{equation}
where ${v _{i,m}^{{\mathrm{L}},k}}$ and ${v_{i,m}^{{\mathrm{U}},k}}$ denote the lower bound and upper bound for ${{v _{i,m}}}$ in the $k$-th iteration. Index $\left({{i^k},{m^k}}\right)$ in \eqref{Subdivision} corresponds to the variable with the longest normalized edge in $\Phi^k$, i.e., $\left({{i^k},{m^k}}\right) = \arg \underset{\left(i,m\right)}{\maxo}\; \frac{{ {v_{i,m}^{{\rm{U}},k} - v_{i,m}^{{\rm{L}},k}}}}{{ {v_{i,m}^{{\rm{U}},1} - v_{i,m}^{{\rm{L}},1}} }}$.
Figure \ref{BnB:a} illustrates the successive branching procedure in a two-dimensional space with ${N_{\mathrm{F}}}=1$ and $M=1$.
Firstly, we partitioned the initial hyper-rectangle $\Phi^1$ into $\Phi^{1,1}$ and $\Phi^{1,2}$ via perform a bisection on the edge of $v_{1,1}$.
Then, $\Phi^{1,2}$ is selected as current hyper-rectangle in the 2-th iteration, namely $\Phi^2$, for subsequent branching iterations.
The shadowed region denotes the current hyper-rectangle in the $k$-th iteration, i.e., $\Phi^k$.
Note that the branching procedure is exhaustive due to the finite numbers of optimization variables and the finite volume of initial hyper-rectangle $\Phi^1$, i.e., $\underset{k \to \infty }{\lim}\underset{\left(i,m\right)}{\maxo}\; \left( v _{i,m}^{{\mathrm{U}},k}-v _{i,m}^{{\mathrm{L}},k}\right) = 0.$ Correspondingly, in Figure \ref{BnB:a}, the shadowed region will collapse into a point with $k \to \infty$.

Now, in the $k$-th iteration, we can rewrite the problem in \eqref{P3Penalty} within $\Phi^k$ as follows:
\vspace*{-4mm}
\begin{equation}\label{P3PenaltySubrect} \underset{\left(\mathbf{\overline{s}},\;{\boldsymbol{{\gamma} }} \right) \in \Phi^k}{\mino}\,\, \,\, G^{\theta}\left(\mathbf{\overline{s}},{\boldsymbol{{\gamma} }} \right)\;\;\;
\mbox{s.t.} \;\; \widetilde{\text{\textbf{C6}}},\;\overline{\text{\textbf{C7}}}.
\vspace*{-2mm}
\end{equation}

%where $\left(\mathbf{\overline{S}},{\boldsymbol{{\gamma} }} \right) \in \Phi^k$ denotes ${{\gamma}_{i,m}^{\mathrm{L},k} = x _{i,m}^{{\mathrm{L}},k}}$, ${{\gamma} _{i,m}^{\mathrm{U},k} = x _{i,m}^{{\mathrm{U}},k}}$, ${{\overline{s}}_{i,m}^{\mathrm{L},k} = x _{i,m+M}^{{\mathrm{L}},k}}$, and ${{ \overline{s}}_{i,m}^{\mathrm{U},k} = x _{i,m+M}^{{\mathrm{U}},k}}$.
\subsubsection{Lower Bound and Upper Bound}
We first present the lower bound for the problem in \eqref{P3PenaltySubrect}, from which the upper bound can be obtained straightforwardly in the end of this part. As it was shown in \cite{Androulakis1995}, the tightest possible convex lower bound of a product term $a_{i,m,n}{{v_{i,m}v_{i,n}}}$ inside some rectangular region $D_{i,m,n}^k = \left[ v _{i,m}^{{\mathrm{L}},k},v _{i,m}^{{\mathrm{U}},k} \right] \times \left[ v _{i,n}^{{\mathrm{L}},k},v _{i,n}^{{\mathrm{U}},k} \right]$, i.e., convex envelope, in the $k$-th iteration is given by
\vspace*{-4mm}
\begin{align}\label{ConvexEnvelope}
&{l_{a_{i,m,n}}^k}\left( {{v _{i,m}},{v _{i,n}}} \right) = \notag \\[-2mm]
&\left\{
\begin{array}{ll}
\hspace*{-2mm}a_{i,m,n} \max\left( {v _{i,m}^{{\mathrm{L}},k}{v _{i,n}} \hspace*{-1mm}+\hspace*{-1mm} v _{i,n}^{{\mathrm{L}},k}{v _{i,m}} \hspace*{-1mm}-\hspace*{-1mm} v _{i,m}^{{\mathrm{L}},k}v _{i,n}^{{\mathrm{L}},k},\;x_{i,m}^{{\mathrm{U}},k}{x _{i,n}} \hspace*{-1mm}+\hspace*{-1mm} v _{i,n}^{{\mathrm{U}},k}{v _{i,m}} \hspace*{-1mm}-\hspace*{-1mm} v _{i,m}^{{\mathrm{U}},k} v _{i,n}^{{\mathrm{U}},k}} \right) & \text{if}\;{a_{i,m,n}} > 0,\\
\hspace*{-2mm}a_{i,m,n} \min\left( {v _{i,m}^{{\mathrm{L}},k}{v _{i,n}} \hspace*{-1mm}+\hspace*{-1mm} v _{i,n}^{{\mathrm{U}},k}{v _{i,m}} \hspace*{-1mm}-\hspace*{-1mm} v _{i,m}^{{\mathrm{L}},k}v _{i,n}^{{\mathrm{U}},k},\;v_{i,m}^{{\mathrm{U}},k}{v _{i,n}} \hspace*{-1mm}+\hspace*{-1mm} v _{i,n}^{{\mathrm{L}},k}{v _{i,m}} \hspace*{-1mm}-\hspace*{-1mm} v _{i,m}^{{\mathrm{U}},k} v _{i,n}^{{\mathrm{L}},k}} \right) & \text{if}\;{a_{i,m,n}} \le 0.
\end{array}
\right.
\end{align}
\par
\vspace*{-2mm}
\noindent
Note that ${l_{a_{i,m,n}}^k}\left( {{v _{i,m}},{v _{i,n}}} \right)$ is a pointwise linear function, which serves as a convex lower bound for $a_{i,m,n}{{v_{i,m}v_{i,n}}}$ in $D_{i,m,n}^k$, i.e., ${l_{a_{i,m,n}}^k}\left( {{v _{i,m}},{v _{i,n}}} \right) \le a_{i,m,n}{{v_{i,m}v_{i,n}}}$.
% and ${l_{a_{i,m,n}}^k}\left( {{x _{i,m}},{x _{i,n}}} \right) = a_{i,m,n}{{x_{i,m}x_{i,n}}}$ holds on the boundary of the rectangular region.
Further, the maximum separation between $a_{i,m,n}{{v_{i,m}v_{i,n}}}$ and ${l_{a_{i,m,n}}^k}\left( {{v _{i,m}},{v _{i,n}}} \right)$ is equal to one-fourth of the area of rectangular region $D_{i,m,n}^k$ \cite{Androulakis1995}, which is
\vspace*{-2mm}
\begin{equation}\label{MaxSeparation1}
{{\varepsilon}^k} \left( v_{i,m}, v _{i,n} \right) = \frac{1}{4}{\left( {v _{i,m}^{{\mathrm{U}},k} - v _{i,m}^{{\mathrm{L}},k}} \right)\left( {v _{i,n}^{{\mathrm{U}},k} - v _{i,n}^{{\mathrm{L}},k}} \right)}.
\vspace*{-2mm}
\end{equation}
We note that the maximum separation is bounded and ${l_{a_{i,m,n}}^k}\left( {{v _{i,m}},{v _{i,n}}} \right)$ can be arbitrarily close to $a_{i,m,n}{{v_{i,m}v_{i,n}}}$ for a small enough rectangle $D_{i,m,n}^k$.

Based on the convex envelope for product terms in \eqref{P3PenaltySubrect}, in the $k$-th iteration, we have the relaxed convex minimization problem over current hyper-rectangle $\Phi^k$ as follows:
\vspace*{-3.5mm}
\begin{equation}\label{P3PenaltySubrectConvexLowerBound} \underset{\left(\mathbf{\overline{s}},\;\boldsymbol{{\gamma}}\right) \in \Phi^k }{\mino}\;\underline{G}^{\theta}_k\left(\mathbf{\overline{s}},{\boldsymbol{{\gamma} }} \right)\;\;\;\mbox{s.t.} \;\;\widetilde{\text{\textbf{C6}}},\;\overline{\text{\textbf{C7}}},
\vspace*{-3mm}
\end{equation}
where $\underline{G}^{\theta}_k\left(\mathbf{\overline{s}},{\boldsymbol{{\gamma} }} \right)$ denotes the convex lower bounding function for $G^{\theta}\left(\mathbf{\overline{s}},{\boldsymbol{{\gamma} }} \right)$ in \eqref{P3PenaltySubrect} within $\Phi^k$,
\vspace*{-2mm}
\begin{equation}
\underline{G}^{\theta}_k\hspace*{-1mm}\left(\mathbf{\overline{s}},{\boldsymbol{{\gamma} }} \right) \hspace*{-1mm}=\hspace*{-1mm} \sum\limits_{i = 1}^{{N_{\mathrm{F}}}} \hspace*{-1mm}{\sum\limits_{m = 1}^M \hspace*{-1mm}{{{\gamma} _{i,m}}\hspace*{-1mm}\left( {\frac{1}{{{\beta _{i,m}}}} \hspace*{-1mm}+\hspace*{-1mm} \theta } \right)} }  + \sum\limits_{i = 1}^{{N_{\mathrm{F}}}} \hspace*{-1mm}{\sum\limits_{m = 1}^{M - 1} \hspace*{-1mm}{\sum\limits_{n = m + 1}^M \hspace*{-1mm}{\frac{{{l_1^k}\left( {{{\gamma} _{i,m}},{{\gamma} _{i,n}}} \right)}}{{\max \left( {{\beta _{i,m}},{\beta _{i,n}}} \right)}}} } }  + \theta \hspace*{-1mm} \sum\limits_{i = 1}^{{N_{\mathrm{F}}}} \hspace*{-1mm}{\sum\limits_{m = 1}^M {{l_{ - 1}^k}\left( {{{\overline{s}}_{i,m}}{\rm{,}}{{\gamma} _{i,m}}} \right)} }.
\vspace*{-1.5mm}
\end{equation}
According to \eqref{MaxSeparation1}, the maximum gap between $G^{\theta}\left(\mathbf{\overline{s}},{\boldsymbol{{\gamma} }} \right)$ and $\underline{G}^{\theta}_k\left(\mathbf{\overline{s}},{\boldsymbol{{\gamma} }} \right)$ within $\Phi^k$ is given by
\vspace*{-1mm}
\begin{equation}\label{MaxGap}
{\Delta _{\max }^k} = \sum\limits_{i = 1}^{{N_{\mathrm{F}}}} {\sum\limits_{m = 1}^{M - 1} {\sum\limits_{n = m + 1}^M {\frac{{{\varepsilon^k}\left( {{\gamma} _{i,m},{\gamma} _{i,n}} \right)}}{{\max \left( {{\beta _{i,m}},{\beta _{i,n}}} \right)}}} } }  + \theta \sum\limits_{i = 1}^{{N_{\mathrm{F}}}} {\sum\limits_{m = 1}^M {{\varepsilon^k}\left( { \overline{s}_{i,m},{\gamma} _{i,m}} \right)} },
\vspace*{-1mm}
\end{equation}
which will vanish when $\Phi^k$ collapses into a point with $k \to \infty$.
%where ${{\varepsilon _1^k}\left( {{\gamma} _{i,m}^{\mathrm{mid},k},{\gamma} _{i,n}^{\mathrm{mid},k}} \right)}$ and ${{\varepsilon _{ - 1}^k}\left( { \overline{s}_{i,m}^{\mathrm{mid},k},{\gamma} _{i,m}^{\mathrm{mid},k}} \right)}$ follow the definition in \eqref{MaxSeparation1}.
Now, the relaxed problem in \eqref{P3PenaltySubrectConvexLowerBound} is a convex programming problem that can be solved efficiently by standard convex program solvers such as CVX \cite{cvx}.
Note that the optimal value of \eqref{P3PenaltySubrectConvexLowerBound} provides a lower bound for \eqref{P3PenaltySubrect} within $\Phi^k$ locally and that \eqref{P3PenaltySubrect} is infeasible if \eqref{P3PenaltySubrectConvexLowerBound} is infeasible.

For the upper bound, it is clear that any feasible solution of the problem in \eqref{P3PenaltySubrect} attains a local upper bound within the hyper-rectangle $\Phi^k$. It can be observed that the optimal solution of the problem in \eqref{P3PenaltySubrectConvexLowerBound}, denoted as $\left(\mathbf{\overline{s}}^*_k, {\boldsymbol{{\gamma} }^*_k}\right)$, is always a feasible point for \eqref{P3PenaltySubrect}, since they share the same feasible solution set. Therefore, an upper bound of \eqref{P3PenaltySubrect} within $\Phi^k$ can be obtained by simply calculating $G^{\theta}\left(\mathbf{\overline{s}}^*_k,{\boldsymbol{{\gamma} }^*_k} \right)$.

\begin{table}[t]
\vspace*{-10mm}
\begin{algorithm} [H]                    % enter the algorithm environment
\caption{Optimal Resource Allocation Algorithm via B\&B}     % give the algorithm a caption
\label{alg1}                             % and a label for \ref{} commands later in the document
\begin{algorithmic} [1]
\footnotesize          % enter the algorithmic environment
\STATE \textbf{Initialization:}
  Set the convergence tolerance $\epsilon$, the iteration counter $k=1$, and initialize current subrectangle\footnotemark{} $\Phi^k$ through \eqref{InitialRectangle}. Solve the problem in \eqref{P3PenaltySubrectConvexLowerBound} within $\Phi^k$ to obtain the intermediate optimal solution $\left(\mathbf{\overline{s}}_k,{\boldsymbol{{\gamma} }_k}\right)$. Then, define the current point\footnotemark{} $\left(\mathbf{\overline{s}}_k^{\mathrm{cur}},{\boldsymbol{{\gamma} }_k^{\mathrm{cur}}}\right) = \left(\mathbf{\overline{s}}_k,{\boldsymbol{{\gamma} }_k}\right)$, the incumbent point\footnotemark{} $\left(\mathbf{\overline{s}}_k^{\mathrm{inc}},{\boldsymbol{{\gamma} }_k^{\mathrm{inc}}}\right) = \left(\mathbf{\overline{s}}_k,{\boldsymbol{{\gamma} }_k}\right)$, the local lower bound $L_k = \underline{G}^{\theta}_k\left(\mathbf{\overline{s}}_k,{\boldsymbol{{\gamma} }_k} \right)$, and the local upper bound $U_k = G^{\theta}\left(\mathbf{\overline{s}}_k,{\boldsymbol{{\gamma} }_k} \right)$. Initialize the global lower bound and upper bound with $\mathrm{LBD}_{k} = L_k$ and $\mathrm{UBD}_{k} = U_k$, respectively. Initialize the unfathomed partition set with $\mathcal{Z} = \left\{ \Phi^k  \right\}$. Correspondingly, define the local lower bound set and the local upper bound set for all unfathomed rectangles in $\mathcal{Z}$ with $\mathcal{W}$ and $\mathcal{V}$, respectively. Initialize them with $\mathcal{W} = \left\{L_k\right\}$ and $\mathcal{V} = \left\{U_k\right\}$, respectively.

\STATE \textbf{Branching on Current Rectangle:}
  Partition current rectangle $\Phi^k$ into two subrectangles $\Phi^{k,1}$ and $\Phi^{k,2}$ through \eqref{Subdivision} and update the unfathomed partition set with $\mathcal{Z} = \mathcal{Z} \bigcup \left\{\Phi^{k,1},\Phi^{k,2}\right\} \setminus \Phi^k$.

\STATE \textbf{Local Lower Bound and Upper Bound:}
  Solve the problem in \eqref{P3PenaltySubrectConvexLowerBound} within $\Phi^{k,r}$ ($r = 1,2$). If it is infeasible, delete (fathoming) $\Phi^{k,r}$ from $\mathcal{Z}$. Otherwise, we obtain the intermediate optimal solution, $\left(\mathbf{\overline{s}}_{k,r},{\boldsymbol{{\gamma} }_{k,r}} \right)$, and the local lower bound and upper bound within $\Phi^{k,r}$ given by
  $L_{k,r} = \underline{G}^{\theta}_{k,r}\left(\mathbf{\overline{s}}_{k,r},{\boldsymbol{{\gamma} }_{k,r}} \right)$ and
  $U_{k,r} = G^{\theta}\left(\mathbf{\overline{s}}_{k,r},{\boldsymbol{{\gamma} }_{k,r}} \right)$, respectively.

\STATE \textbf{Update $\mathcal{Z}$, $\mathcal{W}$, $\mathcal{V}$, $\mathrm{LBD}_k$, and $\mathrm{UBD}_k$:}
  For $L_{k,r} \ge \mathrm{UBD}_k$, delete (fathoming) $\Phi^{k,r}$ from $\mathcal{Z}$.
  If $\mathcal{Z} = \emptyset$, then stop and return the incumbent point. Update parameters with
    \vspace*{-3mm}
    \begin{eqnarray}\label{UpdateParameters}
    \hspace*{-10mm}&&\mathcal{W} = \mathcal{W} \bigcup L_{k,r} \; \text{and} \;\mathcal{V} = \mathcal{V} \bigcup U_{k,r},\; \text{if} \; \Phi^{k,r} \in \mathcal{Z},\\
    \hspace*{-10mm}&&\mathrm{LBD}_{k+1} = L_{k^*,r^*} = \underset{k',r'}{\mino}\;\left(\mathcal{W}\right)\; \text{and} \;
    \mathrm{UBD}_{k+1} = U_{k^\circ,r^\circ} = \underset{k',r'}{\mino}\;\left(\mathcal{V}\right),\; r'\in \left\{1,2\right\},\; k'\in \left\{1, \ldots , k\right\}.\label{GlobalUpperBound}
    \end{eqnarray}
    \vspace*{-8mm}
\STATE \textbf{Update Current Point, Incumbent Point, and Current Rectangle:}
   $k = k+1$. Update the current point with     $\left(\mathbf{\overline{s}}_{k}^{\mathrm{cur}},{\boldsymbol{{\gamma} }_{k}^{\mathrm{cur}}} \right) = \left(\mathbf{\overline{s}}_{k^*,r^*},{\boldsymbol{{\gamma} }_{k^*,r^*}} \right)$ and the incumbent point with
    $\left(\mathbf{\overline{s}}_{k}^{\mathrm{inc}},{\boldsymbol{{\gamma} }_{k}^{\mathrm{inc}}} \right) = \left(\mathbf{\overline{s}}_{k^\circ,r^\circ},{\boldsymbol{{\gamma} }_{k^\circ,r^\circ}} \right)$,
    where $\left(k^*,r^*\right)$ and $\left(k^\circ,r^\circ\right)$ are obtained from \eqref{GlobalUpperBound}.
    Correspondingly, current rectangle is selected as $\Phi^k = \Phi^{k^*,r^*}$.

\STATE \textbf{Convergence Check:}
    If $\left(\mathrm{UBD}_k-\mathrm{LBD}_k\right)>\epsilon$, then go to \textbf{Step} 2.
    Otherwise, ${\epsilon}\text{-}$convergence solution has been attained and return the incumbent point.
\end{algorithmic}
\end{algorithm}\vspace*{-17mm}

\end{table}
\footnotetext[9]{In our algorithm, the current subrectangle is the one which possesses the minimum local lower bound among all the unfathomed subrectangles.}
\footnotetext[10]{The current point denotes the intermediate optimal solution within the current subrectangle.}
\footnotetext[11]{The incumbent point is the best feasible solution that we have found up to current iteration. It is an intermediate optimal solution within some subrectangle which possesses the minimum local upper bound among all the unfathomed subrectangles.\vspace*{-10mm}}

\subsubsection{Overall Algorithm}
Based on the proposed branching procedure and bounding
methods, we develop the B\&B resource allocation algorithm to obtain the globally optimal solution for the problem in \eqref{P3Penalty}, cf. \textbf{Algorithm} \ref{alg1}.
Accordingly, Figure \ref{BnB:b} illustrates a simple example of the developed algorithm in a one-dimensional space, where $\underline{G}^{\theta}_{k,r}\left(\mathbf{\overline{s}},{\boldsymbol{{\gamma} }} \right)$, $r = 1,2$, denotes the lower bounding function for $G^{\theta}\left(\mathbf{\overline{s}},{\boldsymbol{{\gamma} }} \right)$ in \eqref{P3PenaltySubrect} within $\Phi^{k,r}$.
In Step 4, if $L_{k,r} \ge \mathrm{UBD}_k$, the optimal solution must not locate in $\Phi^{k,r}$, and thus we discard it from $\mathcal{Z}$, such as $\Phi^{k,1}$ in Figure \ref{BnB:b}.
Note that the lower bound within subrectangle $\Phi^{k,r}$ is always larger than that within $\Phi^{k}$ since the feasible solution set becomes smaller, i.e., $L_{k,r} \ge L_{k}$.
Therefore, the global lower bound update and current rectangle selection operation in Steps 4 and 5 can generate a \emph{non-decreasing} sequence for $\mathrm{LBD}_k$.
On the other hand, the global upper bound update operation can generate a \emph{non-increasing} sequence for $\mathrm{UBD}_k$.
For example, in Figure \ref{BnB:b}, we can easily observe that $\mathrm{UBD}_{k+1}\le\mathrm{UBD}_k$ and $\mathrm{LBD}_{k+1}\ge\mathrm{LBD}_k$.
It can be proved that the proposed branch and bound algorithm converges to the globally optimal solution in finite number of iterations based on the sufficient conditions stated in \cite{horst2013global}.
The proof of convergence for the adopted B\&B algorithm can be found in \cite{MARANASProofBB}.
The convergence speed of our proposed algorithm will be verified by simulations in Section \ref{SimulationResults}.
\begin{figure}[t]
\centering
\vspace*{-8mm}
\includegraphics[width=3.5in]{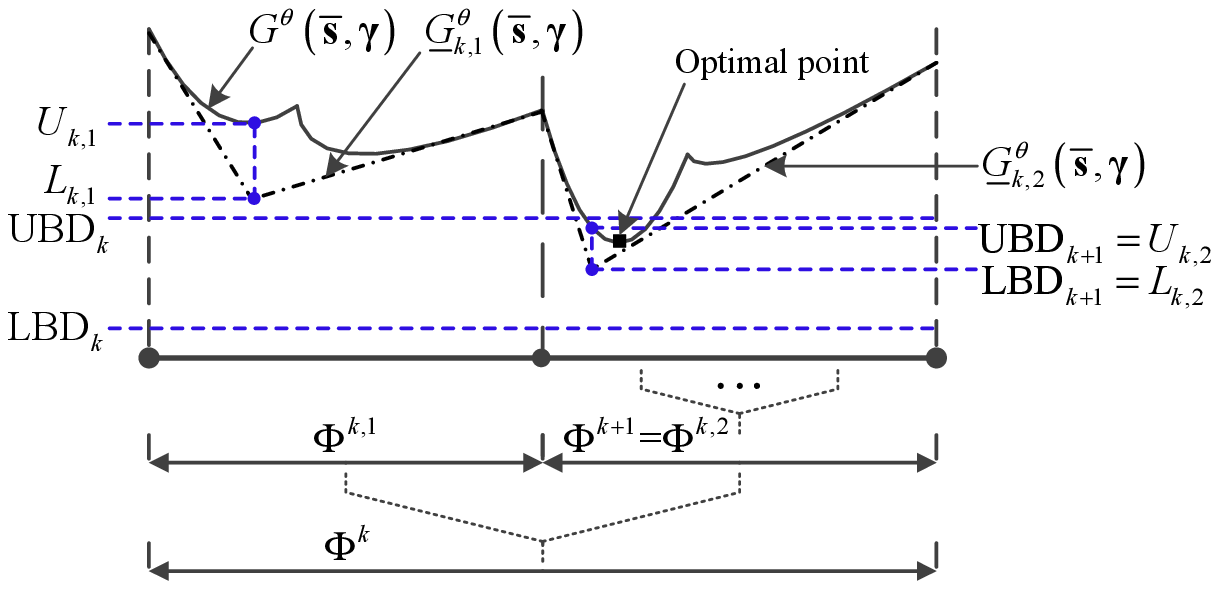}\vspace*{-6mm}
\caption{An illustration of \textbf{Algorithm} \ref{alg1} from the $k$-th iteration to the $(k+1)$-th iteration in a one-dimensional space.}\vspace*{-10mm}
\label{BnB:b}
\end{figure}
\vspace*{-4mm}
\section{Suboptimal Solution} \label{Suboptimal}
Compared to the brute-force search method, the proposed B\&B algorithm provides a systematic approach by exploiting the structure of the problem in \eqref{P3Penalty}.
It saves a large amount of computational complexity since it discards the non-promising subrectangles\cite{horst2013global}.
More importantly, it serves as a performance benchmark for any suboptimal algorithm.
However, it has a non-polynomial time computational complexity\cite{Konno2000}.
In this section, we present a suboptimal solution for the problem in \eqref{P2} by exploiting the D.C. programming\cite{dinh2010local}, which only requires a polynomial time computational complexity.

To start with, we aim to circumvent the coupling between binary variables ${s}_{i,m}$ and continuous variables ${\gamma} _{i,m}$ in \eqref{P2}. We define an auxiliary variable $\widetilde{\gamma}_{i,m} = {\gamma} _{i,m}{s}_{i,m}$ and adopt the big-M formulation\cite{DerrickFD2016} to equivalently transform the problem in \eqref{P2} as follows:
\vspace*{-2mm}
\begin{align}\label{P4}
&\underset{\mathbf{s},\boldsymbol{\gamma},\widetilde{\boldsymbol{\gamma}}}{\mino}\,\, \,\, \notag \sum\limits_{i = 1}^{{N_{\mathrm{F}}}} {\sum\limits_{m = 1}^M {\frac{{{\widetilde{\gamma} _{i,m}}}}{{{\beta _{i,m}}}}} }  + \sum\limits_{i = 1}^{{N_{\mathrm{F}}}} {\sum\limits_{m = 1}^{M-1} {\sum\limits_{n = m + 1}^M {\frac{{{\widetilde{\gamma} _{i,m}}{\widetilde{\gamma} _{i,n}}}}{{\max \left( {{\beta _{i,m}},{\beta _{i,n}}} \right)}}} } } \\[-3mm]
\notag\mbox{s.t.}\;\;
%%%%%
&\mbox{\textbf{C1}},\;\mbox{\textbf{C4}},\;\mbox{\textbf{C7}},\notag\\[-3mm]
&\mbox{\textbf{C6}: } \sum\limits_{i = 1}^{{N_{\mathrm{F}}}} {s_{i,m}}{{{\log }_2}\hspace*{-1mm}\left( {1 \hspace*{-1mm}+\hspace*{-1mm}\frac{\widetilde{\gamma} _{i,m}}{{s_{i,m}}}} \right)\hspace*{-1mm}}  \ge {{R}_m^{\mathrm{total}}},\;\forall m,
\;\mbox{\textbf{C9}: } {{\widetilde {\gamma} }_{i,m}} \hspace*{-1mm}\ge\hspace*{-1mm} 0,\;\forall i,m,
\;\mbox{\textbf{C10}: } {{\widetilde {\gamma}}_{i,m}} \hspace*{-1mm}\le\hspace*{-1mm} {\gamma _{i,m}},\;\forall i,m,\notag\\[-2mm]
&\mbox{\textbf{C11}: } {{\widetilde {\gamma} }_{i,m}} \le {s_{i,m}} \hspace*{-1mm}\left(2^{R_m^{{\mathrm{total}}}}\hspace*{-1mm}-\hspace*{-1mm}1\right)\hspace*{-1mm},
\;\forall i,m,
\;\mbox{\textbf{C12}: } {{\widetilde {\gamma} }_{i,m}} \ge {\gamma _{i,m}} \hspace*{-1mm}-\hspace*{-1mm} \left( {1 \hspace*{-1mm}-\hspace*{-1mm} {s_{i,m}}} \right)\hspace*{-1mm} \left(2^{R_m^{{\mathrm{total}}}}\hspace*{-1mm}-\hspace*{-1mm}1\right),\;\forall i,m,
\end{align}
\par
\vspace*{-3mm}
\noindent
where $\widetilde{\boldsymbol{\gamma}} \in \mathbb{R}^{N_{\mathrm{F}}M \times 1}$ denotes the set of the auxiliary variables $\widetilde{\gamma}_{i,m}$ and constraints \textbf{C9}-\textbf{C12} are imposed additionally following the big-M formulation\cite{DerrickFD2016}.
Besides, the binary constraints in \textbf{C1} are another major obstacle for the design of a computationally efficient resource allocation algorithm. Hence, we rewrite the binary constraint \textbf{C1} in its equivalent form:
\vspace*{-2mm}
\begin{equation}\label{C1AB}
\mbox{\textbf{C1a}: }0 \le {s}_{i,m} \le 1 \quad \text{and} \quad
\mbox{\textbf{C1b}: }\sum\limits_{i = 1}^{{N_\mathrm{F}}} {\sum\limits_{m = 1}^M {{{s}_{i,m}}} }  - \sum\limits_{i = 1}^{{N_\mathrm{F}}} {\sum\limits_{m = 1}^M {s_{i,m}^2} }  \le 0, \;\;\forall i,m.
\vspace*{-2mm}
\end{equation}

Furthermore, we rewrite ${{\widetilde{\gamma} _{i,m}}{\widetilde{\gamma} _{i,n}}} = \frac{1}{2}{{\left( {{{\widetilde \gamma }_{i,m}} + {{\widetilde \gamma }_{i,n}}} \right)}^2} - \frac{1}{2}{\left( {\widetilde \gamma _{i,m}^2 + \widetilde \gamma _{i,n}^2} \right)}$ and augment the D.C. constraint \textbf{C1b} into the objective function via a penalty factor $\eta \gg 1$. The problem in \eqref{P4} can be rewritten in the canonical form of D.C. programming as follows:
\vspace*{-4mm}
\begin{equation}\label{P4DC}
\underset{\mathbf{s},\boldsymbol{\gamma},\widetilde{\boldsymbol{\gamma}}}{\mino}\,\, \,\, G_{1}^{\eta}\left(\mathbf{{s}},{\boldsymbol{\widetilde{\gamma} }} \right) - G_{2}^{\eta}\left(\mathbf{{s}},{\boldsymbol{\widetilde{\gamma} }} \right) \;\;\;
\mbox{s.t.}
%%%%%
\;\;\mbox{\textbf{C1a}},\;\mbox{\textbf{C4}},\;\mbox{\textbf{C6}},\;\mbox{\textbf{C7}},\;\mbox{\textbf{C9}-\textbf{C12}},
\vspace*{-3mm}
\end{equation}
where
\vspace*{-2mm}
\begin{align}
G_{1}^{\eta}\left(\mathbf{{s}},{\boldsymbol{\widetilde{\gamma} }} \right) &= \sum\limits_{i = 1}^{{N_\mathrm{F}}} {\sum\limits_{m = 1}^M {\frac{{{{\widetilde \gamma }_{i,m}}}}{{{\beta _{i,m}}}}} }  + \frac{1}{2}\sum\limits_{i = 1}^{{N_\mathrm{F}}} {\sum\limits_{m = 1}^{M - 1} {\sum\limits_{n = m + 1}^M {\frac{{{{\left( {{{\widetilde \gamma }_{i,m}} + {{\widetilde \gamma }_{i,n}}} \right)}^2}}}{{\max \left( {{\beta _{i,m}},{\beta _{i,n}}} \right)}}} } }  + \eta \sum\limits_{i = 1}^{{N_\mathrm{F}}} {\sum\limits_{m = 1}^M {{s_{i,m}}} } ,\label{DC1}\\[-1mm]
G_{2}^{\eta}\left(\mathbf{{s}},{\boldsymbol{\widetilde{\gamma} }} \right) &= \frac{1}{2}\sum\limits_{i = 1}^{{N_\mathrm{F}}} {\sum\limits_{m = 1}^{M - 1} {\sum\limits_{n = m + 1}^M {\frac{{\left( {\widetilde \gamma _{i,m}^2 + \widetilde \gamma _{i,n}^2} \right)}}{{\max \left( {{\beta _{i,m}},{\beta _{i,n}}} \right)}}} } }  + \eta \sum\limits_{i = 1}^{{N_\mathrm{F}}} {\sum\limits_{m = 1}^M {s_{i,m}^2} } \label{DC2}.
\end{align}
\par\vspace*{-2mm}

According to Theorem \ref{Theorem3}, the problem in \eqref{P4DC} is equivalent to the problem in \eqref{P4}. Note that ${G_{1}^{\eta}} \left(\mathbf{{s}},{\boldsymbol{\widetilde{\gamma} }} \right)$ and ${G_{2}^{\eta}} \left(\mathbf{{s}},{\boldsymbol{\widetilde{\gamma} }} \right)$ are differentiable convex functions w.r.t. $s_{i,m}$ and ${{\widetilde \gamma }_{i,m}}$. Therefore, for any feasible point $\left(\mathbf{{s}}_{k},{\boldsymbol{\widetilde{\gamma} }}_{k} \right)$, we can define the global underestimator for ${G_{2}^{\eta}} \left(\mathbf{{s}},{\boldsymbol{\widetilde{\gamma} }} \right)$ based on its first order Taylor's expansion at $\left(\mathbf{{s}}_{k},{\boldsymbol{\widetilde{\gamma} }}_{k} \right)$ as follows:
\vspace*{-3.5mm}
\begin{equation}\label{Taylor1}
G_2^\eta\left( {{\mathbf{{s}}},{\boldsymbol{\widetilde{\gamma} }}} \right) \ge G_2^\eta\left( {{{\mathbf{{s}}}_k},{\boldsymbol{\widetilde{\gamma} }_k}} \right) + {\nabla _{{\mathbf{{s}}}}}G_2^\eta{\left( {{{{\mathbf{{s}}}}_k},{{\boldsymbol{\widetilde{\gamma} }}_k}} \right)^{\mathrm{T}}}\left( {{\mathbf{{s}}} - {{{\mathbf{{s}}}}_k}} \right) + {\nabla _{\boldsymbol{\widetilde{\gamma} }}}G_2^\eta{\left( {{{{\mathbf{{s}}}}_k},{{\boldsymbol{\widetilde{\gamma} }}_k}} \right)^{\mathrm{T}}}\left( {{\boldsymbol{\widetilde{\gamma} }} - {{\boldsymbol{\widetilde{\gamma} }}_k}} \right),
\vspace*{-3.5mm}
\end{equation}
where ${\nabla _{{\mathbf{{s}}}}}G_2^\eta{\left( {{{{\mathbf{{s}}}}_k},{{\boldsymbol{\widetilde{\gamma} }}_k}} \right)}$ and ${\nabla _{\boldsymbol{\widetilde{\gamma} }}}G_2^\eta{\left( {{{{\mathbf{{s}}}}_k},{{\boldsymbol{\widetilde{\gamma} }}_k}} \right)}$ denote the gradient vectors of ${G_2^\eta} \left(\mathbf{{s}},{\boldsymbol{\widetilde{\gamma} }} \right)$ at $\left( {{{{\mathbf{{s}}}}_k},{{\boldsymbol{\widetilde{\gamma} }}_k}} \right)$ w.r.t. ${{\mathbf{{s}}}}$ and ${\boldsymbol{\widetilde{\gamma} }}$, respectively. Then, we obtain an upper bound for the problem in \eqref{P4DC} by solving the following convex optimization problem:
\vspace*{-4mm}
\begin{eqnarray}\label{P4DCUpperBound} \hspace*{-5mm}&&\underset{\mathbf{{s}},\boldsymbol{{\gamma}},\boldsymbol{\widetilde{\gamma}}}{\mino}\,\, \notag G_1^\eta\hspace*{-1mm}\left(\mathbf{{s}},{\boldsymbol{\widetilde{\gamma} }} \right) \hspace*{-1mm}-\hspace*{-1mm} G_2^\eta\hspace*{-1mm}\left( {{{\mathbf{{s}}}_k},{\boldsymbol{\widetilde{\gamma} }_k}} \right) \hspace*{-1mm}-\hspace*{-1mm} {\nabla _{{\mathbf{{s}}}}}G_2^\eta\hspace*{-0.5mm}{\left( {{{{\mathbf{{s}}}}_k},{{\boldsymbol{\widetilde{\gamma} }}_k}} \right)^{\mathrm{T}}}\hspace*{-1.5mm}\left( {{\mathbf{{s}}} \hspace*{-1mm}-\hspace*{-1mm} {{{\mathbf{{s}}}}_k}} \right)
\hspace*{-1mm}-\hspace*{-1mm} {\nabla _{\boldsymbol{\widetilde{\gamma} }}}G_2^\eta\hspace*{-0.5mm}{\left( {{{{\mathbf{{s}}}}_k},{{\boldsymbol{\widetilde{\gamma} }}_k}} \right)^{\mathrm{T}}}\hspace*{-1.5mm}\left( {{\boldsymbol{\widetilde{\gamma} }} \hspace*{-1mm}-\hspace*{-1mm} {{\boldsymbol{\widetilde{\gamma} }}_k}} \right)\notag\\[-1mm]
\hspace*{-5mm}\mbox{s.t.}\hspace*{-5mm}
%%%%%
&&\mbox{\textbf{C1a}},\;\mbox{\textbf{C4}},\;\mbox{\textbf{C6}},\;\mbox{\textbf{C7}},\;\mbox{\textbf{C9}-\textbf{C12}},
\end{eqnarray}
\par
\vspace*{-2mm}
\noindent
where ${\nabla _{{\mathbf{{s}}}}}G_2^\eta{\left( {\mathbf{{s}_k},{{\boldsymbol{\widetilde{\gamma} }}_k}} \right)^{\mathrm{T}}}\left( {{\mathbf{{s}}} - {{{\mathbf{{s}}}}_k}} \right) =  2\eta \sum\limits_{i = 1}^{{N_{\mathrm{F}}}} {\sum\limits_{m = 1}^M { {s}_{i,m}^k } \left( {{{{s}}_{i,m}} - {s}_{i,m}^k} \right)}$ and
${\nabla _{\boldsymbol{\widetilde{\gamma} }}}G_2^\eta{\left( {{{{\mathbf{{s}}}}_k},{{\boldsymbol{\widetilde{\gamma} }}_k}} \right)^{\mathrm{T}}} \left( {{\boldsymbol{\widetilde{\gamma} }} - {{\boldsymbol{\widetilde{\gamma} }}_k}} \right) = \\ \sum\limits_{i = 1}^{{N_{\rm{F}}}} {\sum\limits_{m = 1}^M {\sum\limits_{n \ne m}^M {\frac{{\gamma _{i,m}^k\left( {{\gamma _{i,m}} - \gamma _{i,m}^k} \right)}}{{\max \left( {{\beta _{i,m}},{\beta _{i,n}}} \right)}}} } }.$
\begin{table}
\vspace*{-10mm}
\begin{algorithm} [H]                    % enter the algorithm environment
\caption{Suboptimal Resource Allocation Algorithm}     % give the algorithm a caption
\label{alg2}                             % and a label for \ref{} commands later in the document
\begin{algorithmic} [1]
\footnotesize          % enter the algorithmic environment
\STATE \textbf{Initialization}\\
Initialize the convergence tolerance $\epsilon$, the maximum number of iterations $K_\mathrm{max}$, the iteration counter $k = 1$, and the initial feasible solution $\left(\mathbf{{s}}_k,{\boldsymbol{\widetilde{\gamma} }_k} \right)$.

\REPEAT
\STATE Solve \eqref{P4DCUpperBound} for a given $\left(\mathbf{{s}}_k,{\boldsymbol{\widetilde{\gamma} }_k} \right)$ and obtain the intermediate resource allcation policy $\left(\mathbf{{s}}',{\boldsymbol{\widetilde{\gamma} }'} \right)$
\STATE Set $k=k+1$ and $\left(\mathbf{{s}}_k,{\boldsymbol{\widetilde{\gamma} }_k} \right) = \left(\mathbf{{s}}',{\boldsymbol{\widetilde{\gamma} }'} \right)$
\UNTIL
$k = K_\mathrm{max}$ or $\max \left\{ {\left\| {\left(\mathbf{{s}}_k,{\boldsymbol{\widetilde{\gamma} }_k} \right) - \left(\mathbf{{s}}_{k-1},{\boldsymbol{\widetilde{\gamma} }_{k-1}} \right)} \right\|_2} \right\} \le \epsilon$
\STATE Return the optimal solution $\left(\mathbf{{s}}^*,{\boldsymbol{\widetilde{\gamma} }^*} \right) = \left(\mathbf{{s}}_k,{\boldsymbol{\widetilde{\gamma} }_k} \right)$
\end{algorithmic}
\end{algorithm}\vspace*{-18mm}
\end{table}

Now, the problem in \eqref{P4DCUpperBound} is a convex programming problem which can be easily solved by CVX \cite{cvx}.
The solution of the problem in \eqref{P4DCUpperBound} provides an upper bound for the problem in \eqref{P4DC}.
\textcolor[rgb]{0.00,0.00,0.00}{To tighten the obtained upper bound, we employ an iterative algorithm to generate a sequence of feasible solution successively, cf. \textbf{Algorithm} \ref{alg2}.
In \textbf{Algorithm} \ref{alg2}, the initial feasible solution $\left(\mathbf{{s}}_1,{\boldsymbol{\widetilde{\gamma} }_1} \right)$ is obtained via solving the problem in \eqref{P4DC} with $G_{1}^{\eta}\left(\mathbf{{s}},{\boldsymbol{\widetilde{\gamma} }} \right)$
as the objective function\footnote{\textcolor[rgb]{0.00,0.00,0.00}{Note that with $G_{1}^{\eta}\left(\mathbf{{s}},{\boldsymbol{\widetilde{\gamma} }} \right)$
as the objective function, the problem in (41) is a convex problem. Therefore, the initial feasible solution can be found via existing algorithms \cite{Boyd2004} for solving convex problems with a polynomial time computational complexity.}\vspace*{-10mm}}.
The problem in \eqref{P4DCUpperBound} is updated with the intermediate solution from the last iteration and is solved to generate a feasible solution for the next iteration.
Such an iterative procedure will stop when the maximum iteration number is reached or the change of optimization variables is smaller than a predefined convergence tolerance.}
It has been shown that the proposed suboptimal algorithm converges a stationary point with a polynomial time computational complexity for differentiable ${G_{\eta}^{1}} \left(\mathbf{{s}},{\boldsymbol{\widetilde{\gamma} }} \right)$ and ${G_{\eta}^{2}} \left(\mathbf{{s}},{\boldsymbol{\widetilde{\gamma} }} \right)$ \cite{VucicProofDC}.
Noted that there is no guarantee that \textbf{Algorithm} \ref{alg2} can
converge to a globally optimum of the problem in \eqref{P2}. However, our simulation results in the next section will demonstrate its close-to-optimal performance.

\textcolor[rgb]{0.00,0.00,0.00}{In practice, different numerical methods can be used to solve the convex problem in \eqref{P4DCUpperBound} \cite{Boyd2004}.
Particularly, the computational complexity of proposed suboptimal algorithm implemented by the primal-dual path-following interior-point method is \cite{Nemirovski2004IPM}
\vspace*{-2mm}
\begin{equation}
\mathcal{O}\left(K_{\mathrm{max}}\underbrace{\left(\sqrt{8N_{\mathrm{F}}M+N_{\mathrm{F}}+M}\ln(\frac{1}{\Delta})\right)}_{\text{Number of Newton iterations}}
\underbrace{\left( \left(N_{\mathrm{F}}+M\right)(N_{\mathrm{F}}M)^2+35(N_{\mathrm{F}}M)^3\right)}_{\text{Complexity per Newton iteration}}\right),
\end{equation}
% reference Chapter 10 Geometric programming m = 7N_{\mathrm{F}}M+N_{\mathrm{F}}+M   k = N_{\mathrm{F}}M   n = 3N_{\mathrm{F}}M
for a given solution accuracy $\Delta>0$ of the adopted numerical solver, where $\mathcal{O}(\cdot)$ is the big-O notation.
On the other hand, for the proposed optimal algorithm in \textbf{Algorithm} \ref{alg1}, we note that although the B\&B algorithm is guaranteed to find the optimal solution, the required computational complexity in the worst-case is as high as that of an exhaustive search.
The computational complexity of an exhaustive search for the problem in \eqref{P2} is $\mathcal{O}\Bigg(2^{N_{\mathrm{F}}M}\Big(\prod_{m=1}^{M}\frac{2^{R_m^{\mathrm{total}}}-1}{\Delta}\Big)^{N_{\mathrm{F}}}\Bigg),$
for a given solution accuracy $\Delta>0$.
Therefore, the proposed suboptimal algorithm provides a substantial saving in computational complexity compared to the exhaustive search approach.
We note that proposed suboptimal algorithm with a polynomial time computational complexity is desirable for real time implementation \cite{thomas2001introduction}.}

\section{Simulation Results} \label{SimulationResults}
In this section, we evaluate the performance of our proposed resource allocation algorithms through simulations.
Unless specified otherwise, the system parameters used in the our simulations are given as follows.
A single-cell with a BS located at the center with a cell size of $500$ m is considered.
%The number of subcarriers $N_{\mathrm{F}}$ and users $M$ are changing for different simulation cases.
The carrier center frequency is $1.9$ GHz and the bandwidth of each subcarrier is $15$ kHz.
There are $M$ users randomly and uniformly distributed between $30$ m and $500$ m, i.e., $d_i \sim U[30,\;500]$ m, and their target data rates are generated by ${R}_m^{\mathrm{total}} \sim U[1,\;10]$ bit/s/Hz.
The required outage probability of each user on each subcarrier is generated by $\delta_{i,m} \sim U[10^{-5},\;10^{-1}]$.
The user noise power on each subcarrier is $\sigma_{i,m}^2=-128$ dBm and the variance of channel estimation error is $\kappa^2_{i,m} = 0.1$, $\forall i,m$.
The 3GPP urban path loss model with a path loss exponent of $3.6$ \cite{Access2010} is adopted in our simulations.
For our proposed iterative optimal and suboptimal resource allocation algorithms, the maximum error tolerance is set to $\epsilon = 0.01$ and the penalty factors is set to a sufficiently large number such that the value of the penalty term comparable to the value of the objective function\cite{DerrickFD2016}.
The simulations shown in the sequel are obtained by averaging the results over different realizations of different user distances, target data rates, multipath fading coefficients, and outage probability requirements.

For comparison, we consider the performance of the following three baseline schemes.
For baseline 1, the conventional MC-OMA scheme is considered where each subcarrier can only be allocated to at most one user.
To support all $M$ active users and to have a fair comparison, the subcarrier spacing is changed by a factor of $\frac{N_\mathrm{F}}{M}$ to generate $M$ subcarriers.
%Also, the corresponding target rate should be increased by a factor of $\frac{N_\mathrm{F}}{M}$ to satisfy the total transmission rate.
%Then, the minimum power consumption for baseline 1 can be obtained by solving the optimization problem in \eqref{P2} with the new constraint C7: $\sum\limits_{m = 1}^M {{s_{i,m}}} \le 1,\;\forall i$.
Since our proposed scheme subsumes the MC-OMA scheme as a subcase, the minimum power consumption for baseline scheme 1 can be obtained by solving the problem in \eqref{P2} by replacing \textbf{C7} with $\sum\limits_{m = 1}^M {{s_{i,m}}} = 1$.
For baseline 2, a scheme of MC-NOMA with random scheduling is considered where the paired users on each subcarrier is randomly selected\cite{Wei2016NOMA}.
The minimum power consumption for baseline scheme 2 is obtained via solving the problem in \eqref{P4DC} with a given random user scheduling policy $\mathbf{{s}}$ and $\eta = 0$.
For baseline 3, a scheme of MC-NOMA with an equal rate allocation is studied where the target data rate of each user is assigned equally on its allocated subcarriers\cite{Wei2016NOMA}.
%we randomly select $2N_\mathrm{F}-M$ users from $M$ users to halve their target rate on the allocated subcarriers\cite{Wei2016NOMA}. In fact, a system with $N_\mathrm{F}$ subcarriers can only support $2N_\mathrm{F}$ users due to constraint C7 in \eqref{P2} and only at most $2N_\mathrm{F}-M$ users are possible to split their target rate.
Based on the equal rate allocation, the problem in \eqref{P2} can be transformed to a mixed integer linear program, which can be solved by standard numerical integer program solvers, such as Mosek\cite{Mosek2010}, via some non-polynomial time algorithms.
\begin{figure}[t]
\vspace*{-7mm}
\centering
\includegraphics[width=3.5in]{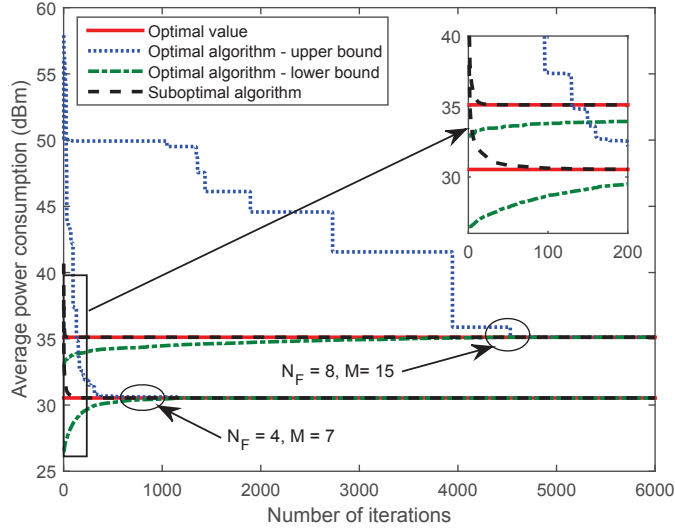}\vspace*{-8mm}
\caption{Convergence of the proposed optimal and suboptimal resource allocation algorithms.}\vspace*{-9mm}
\label{CaseI_Convergence}
\end{figure}
\vspace*{-4mm}
\subsection{Convergence of Proposed Algorithms}
\textcolor[rgb]{0.00,0.00,0.00}{Figure \ref{CaseI_Convergence} illustrates the convergence of our proposed optimal and suboptimal resource allocation algorithms for different values\footnote{Since the computational complexity of the B\&B approach is high, we adopt small values for $M$ and $N_\mathrm{F}$ to compare the gap between the proposed optimal algorithm and the suboptimal algorithm. We note that our proposed suboptimal resource allocation algorithm is computational efficient compared to the optimal one, which can apply to scenarios with more users and subcarriers, such as the simulation case in Section \ref{SimulationD}. In fact, the number of subcarriers in this paper can be viewed as the number of resource blocks in LTE standard\cite{sesia2015lte}, where the user scheduling are performed on resource block level.\vspace*{-10mm}} of $N_{\mathrm{F}}$ and $M$.
For the first case with $N_{\mathrm{F}} = 4$ and $M = 7$, we observe that the optimal algorithm generates a non-increasing upper bound and a non-decreasing lower bound when the number of iterations increases.
Besides, the optimal solution is found when the two bounds meet after $600$ iterations on average.
More importantly, our proposed suboptimal algorithm can converge to the optimal value within $80$ iterations on average.
For the second case with $N_{\mathrm{F}} = 8$ and $M = 15$, it can be observed that the optimal algorithm converges after $4500$ iterations on average.
In fact, the computational complexity of the proposed optimal algorithm increases exponentially w.r.t. the number of optimization variables, and thus the convergence speed is relatively slow for a larger problem size.
However, it can be observed that the suboptimal algorithm converges faster when the numbers of subcarriers and users increase, and it can achieve the optimal value with only 25 iterations in the second case on average.
%This can be attributed to the increasing sparsity of the optimization variables $\mathbf{{s}}$ and $\boldsymbol{{\gamma}}$ in \eqref{P4DCUpperBound} for the second case \cite{DCSparse}.
This is because the time-sharing condition \cite{Yu2006Dual,DerrickLimitedBackhaul} is satisfied with a larger number of subcarriers.
In this case, the optimization problem in \eqref{P4DC} tends to be convexified leading to a higher chance of holding strong duality \cite{DerrickLimitedBackhaul}.
Further, our proposed suboptimal scheme is able to exploit the ``convexity" inherent in large scale optimization problem via the successive convex approximation while the optimal one cannot with relying on feasible set partitioning, cf. Figure \ref{BnB:a}.
Therefore, the proposed suboptimal algorithm converges faster with a larger number of subcarriers in the second case.
To obtain further insight, Table \ref{PairedUsers} shows the solution of our original formulated problem in \eqref{P1} via following the proposed optimal SIC decoding policy for a single channel realization with $N_{\mathrm{F}} = 4$ and $M = 7$ in Figure \ref{CaseI_Convergence}.}
The tick denotes that the user is selected to perform SIC.
It can be observed that two users with distinctive CNR outage thresholds are preferred to be paired together (users 1 and 4, users 3 and 6).
Also, the users (users 2, 4, 6, 7) with higher CNR outage thresholds are selected to perform SIC and only a fraction of power are allocated to them owing to their better channel conditions or non-stringent QoS requirements.
These observations are in analogy to the conclusions for the case of NOMA with perfect CSIT, where users with distinctive channel gains are more likely to be paired, more power is allocated to the weak user, and the strong user is selected to perform SIC\cite{Dingtobepublished}.
Therefore, the defined CNR outage threshold in this paper serves as a metric for determining the optimal SIC decoding policy and resource allocation design for MC-NOMA systems with imperfect CSIT.
\begin{table}[t]
\center \vspace*{-4mm}
\caption{\vspace*{-2mm}Optimal Solution of \eqref{P1} for A Single Channel Realization with $N_{\mathrm{F}} = 4$ and $M = 7$ in Figure \ref{CaseI_Convergence}\vspace*{-2mm}}
  \begin{tabular}{ccccccccc}
  \hline
    Subcarrier index & \multicolumn{2}{c}{1} & \multicolumn{2}{c}{2} & \multicolumn{2}{c}{3} & \multicolumn{2}{c}{4} \\ \hline
    Paired user index & 2 & 5 & 5 & 7 & 1 & 4 & 3 & 6 \\
    Outage threshold $\beta_{i,m}$ & 783.39 & 39.99 & 30.92 & 520.27 & 8.57 & 269.80 & 9.59 & 1349.80 \\
    Rate allocation $R_{i,m}$ (bit/s/Hz)& 8 & 2.03 & 4.97 & 3 & 1 & 7 & 3 & 4\\
%    Power allocation $p_{i,m}$ (mW)& 6.0830 & 0.0698 & 27.9943 & 1.0704 & 0.0325 & 2.1557 & 30.4580 & 0.0128 \\
    Power allocation $p_{i,m}$ (dBm)& 25.13 & 30.35 & 31.42 & 11.29 & 27.69 & 26.73 & 29.07 & 10.46 \\
    SIC decoding   & \Checkmark & - & - & \Checkmark & - & \Checkmark & - & \Checkmark \\
    \hline
  \end{tabular}
\label{PairedUsers}\vspace*{-10mm}
\end{table}
\vspace*{-4mm}
\subsection{Power Consumption versus Target Data Rate}
In Figure \ref{CaseII_TargetRate}, we investigate the power consumption versus the target data rate with $N_{\mathrm{F}} = 8$ and $M = 12$.
In this simulation, all the users have an identical target data rates ${R}_m^{\mathrm{total}}$ and they are set to be from $1$ bit/s/Hz to $10$ bit/s/Hz.
The three baseline schemes are also included for comparison.
As can be observed from Figure \ref{CaseII_TargetRate}, the power consumption increases monotonically with the target data rate for all the schemes.
Clearly, the BS is required to transmit with a higher power to support a more stringent data rate requirement.
Besides, our proposed optimal and suboptimal resource allocation schemes provide a significant power reduction compared to the baseline schemes.
Specifically, baseline scheme 1 requires a higher power consumption (about $3\sim15$ dB) compared to proposed schemes.
This is attributed to the fact that the proposed NOMA schemes are able to distribute the required target data rate of each user over multiple subcarriers efficiently since they admit multiplexing multiple users on each subcarrier.
For OMA schemes, the power consumption increases exponentially with the target data rate requirement since only one subcarrier is allocated to each user in the overloaded scenario.
As a result, the performance gain of NOMA over OMA in terms of power consumption becomes larger when the target data rate increases.
For baseline scheme 2, it can be observed that NOMA is very sensitive to the user scheduling strategy where NOMA with suboptimal random scheduling even consumes more power than that of OMA schemes.
Therefore, a cautiously design of the user scheduling strategy for MC-NOMA systems is fundamentally important in practice.
For baseline scheme 3, the power consumption is slightly higher than that of the proposed schemes but the performance gap is enlarged with an increasing target data rate.
In fact, baseline scheme 3 shares the target data rate equally across the allocated subcarriers of a user, which can realize most of the performance gain of NOMA in low target data rate regimes.
However, our proposed schemes consume less transmit power for high target data rate by exploiting the frequency diversity, where higher rates are allocated to the subcarriers with better channel conditions.
%As a result, in practical resource allocation design for users with low target rate, equal rate allocation for MC-NOMA systems is a simple but effective scheme.
%\begin{figure}[t]
%\centering
%\includegraphics[width=3.5in]{figure/CaseII_TargetRate.eps}\vspace*{-8mm}
%\caption{Power consumption versus target data rate with $N_\mathrm{F} = 8$ and $M = 12$. The power saving gain achieved by the proposed schemes over baseline scheme 1 is denoted by the double-side arrow.}
%\label{CaseII_TargetRate}\vspace*{-8mm}
%\end{figure}
%\vspace*{-2mm}
\begin{figure}[t]
\begin{minipage}{.47\textwidth}
\centering\vspace*{-11mm}
\includegraphics[width=\textwidth]{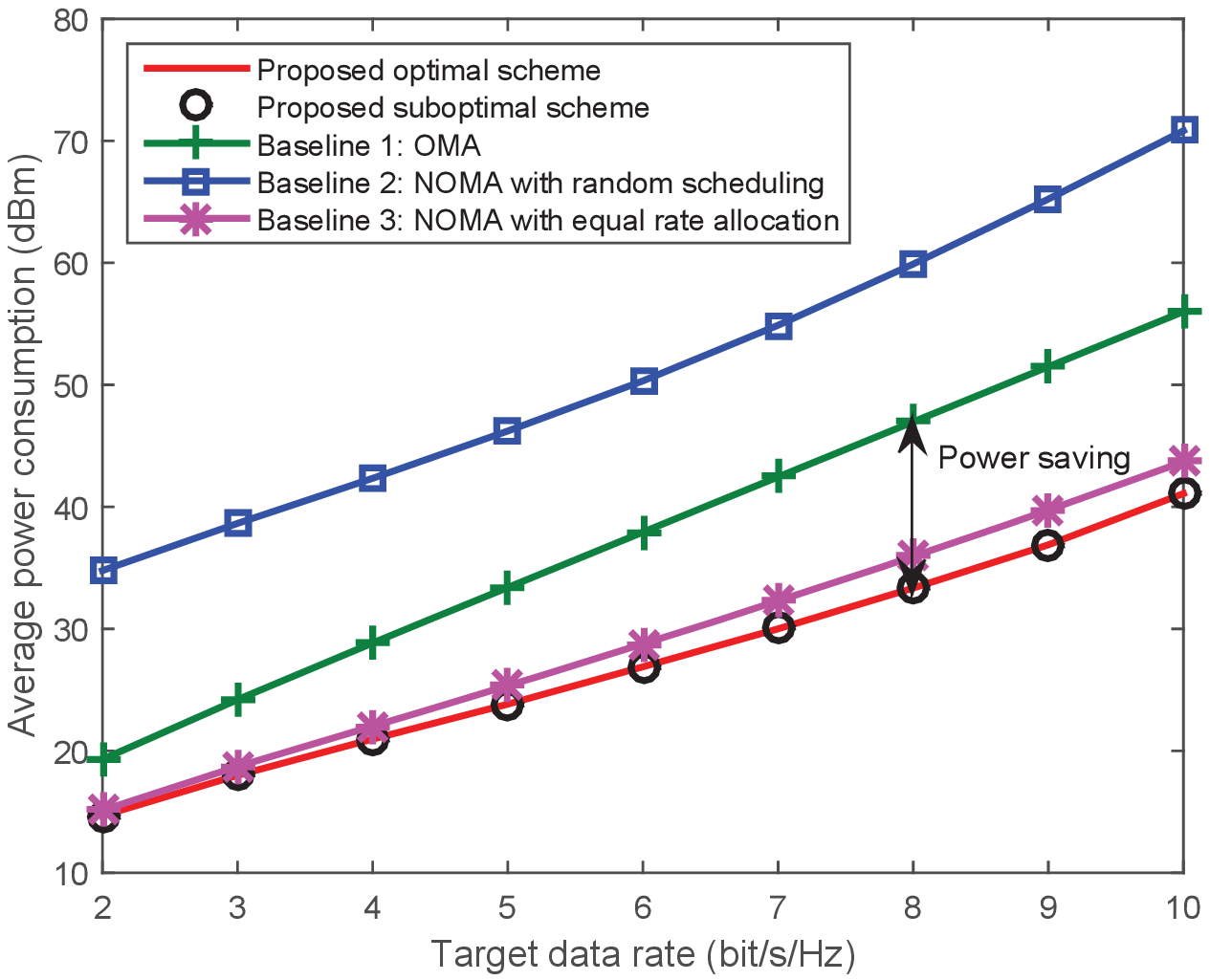}\vspace*{-8mm}
\caption{Average power consumption (dBm) versus target data rate with $N_\mathrm{F} = 8$ and $M = 12$. The power saving gain achieved by the proposed schemes over baseline scheme 1 is denoted by the double-side arrow.}
\label{CaseII_TargetRate}\vspace*{-10mm}
\end{minipage}
\hspace*{1.5mm}
\begin{minipage}{.47\textwidth}
\centering\vspace*{-6mm}
\includegraphics[width=\textwidth]{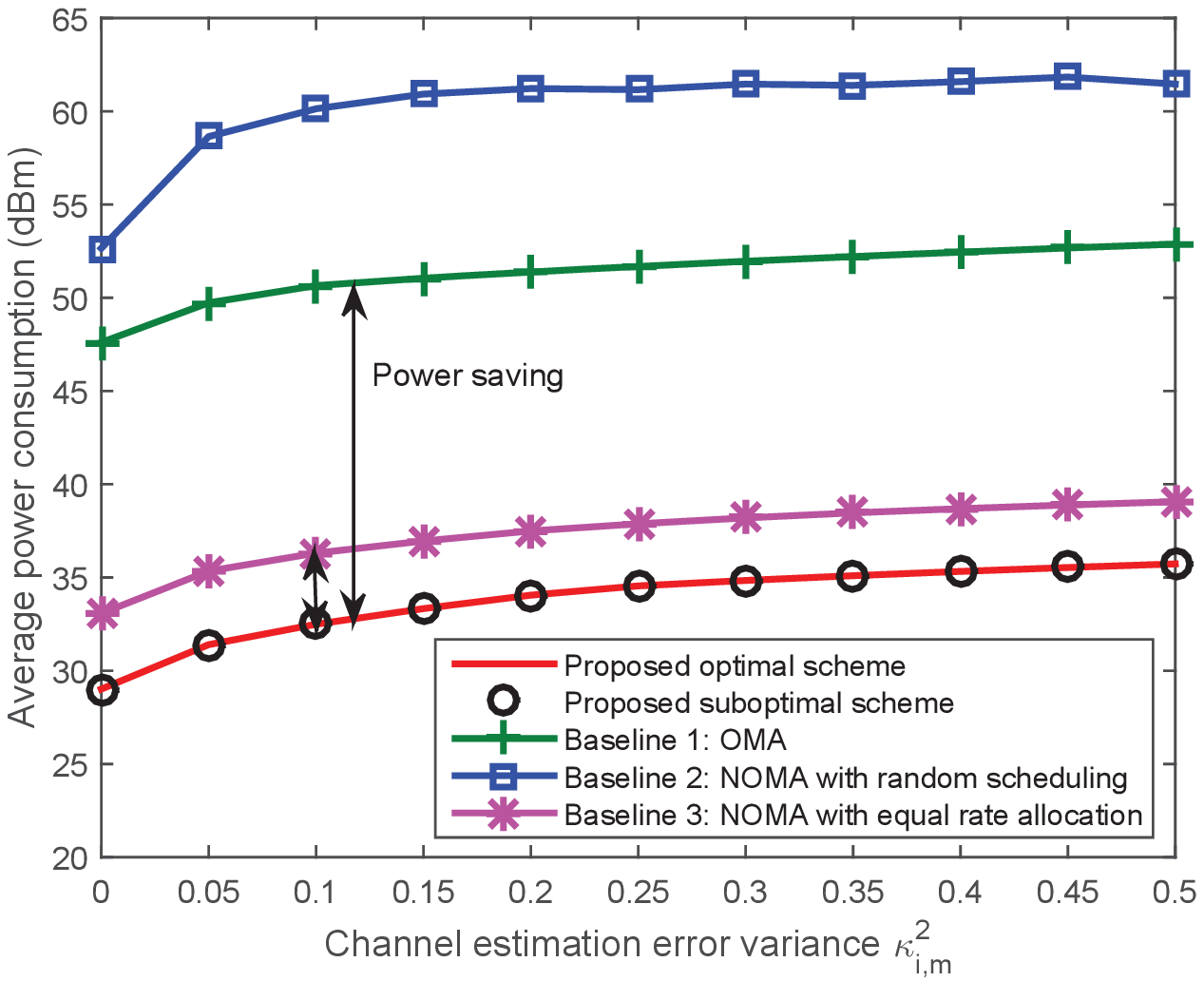}\vspace*{-8mm}
\caption{Average power consumption (dBm) versus channel estimation error variance with $N_\mathrm{F} = 8$ and $M = 12$. The power saving gain achieved by the proposed schemes over baseline scheme 1 and baseline scheme 3 are denoted by the double-side arrows.}
\label{CaseIV_ChannelError}\vspace*{-10mm}
\end{minipage}
\end{figure}
\vspace*{-4mm}
\subsection{Power Consumption versus Channel Estimation Error}
%\begin{figure}[t]
%\vspace*{-8mm}
%\centering
%\includegraphics[width=3.5in]{figure/CaseIV_ChannelError.eps}\vspace*{-7mm}
%\caption{Power consumption versus channel estimation error variance with $N_\mathrm{F} = 8$ and $M = 12$. The power saving gain achieved by the proposed schemes over baseline scheme 1 and baseline scheme 3 are denoted by the double-side arrows.}
%\label{CaseIV_ChannelError}\vspace*{-8mm}
%\end{figure}
Figure \ref{CaseIV_ChannelError} depicts the power consumption versus the variance of channel estimation error with $N_{\mathrm{F}} = 8$, $M = 12$, and ${R}_m^{\mathrm{total}} \sim U[1,\;10]$ bit/s/Hz.
The variance of channel estimation error $\kappa^2_{i,m}$ increasing from $0$ to $0.5$, where $\kappa^2_{i,m} = 0$ denotes that perfect CSIT is available for resource allocation.
It can be observed that the power consumption increases monotonically with $\kappa^2_{i,m}$ for all the schemes.
It is expected that a higher transmit power is necessary to cope with a larger channel uncertainty to satisfy its required outage probability.
Particularly, for our proposed schemes, baseline scheme 1, and baseline scheme 3, a $6$ dB of extra power is required to handle the channel estimation error when $\kappa^2_{i,m}$ increases from $0$ to $0.5$.
However, our proposed schemes are the most power-efficient among all the schemes.
Furthermore, compared to Figure \ref{CaseII_TargetRate} with identical target data rates, the gap of power consumption between baseline scheme 3 and our proposed schemes at $\kappa^2_{i,m} = 0.1$ is enlarged.
Also, the performance gain of our proposed schemes over baseline scheme 1 is larger than that of Figure \ref{CaseII_TargetRate}.
In fact, our proposed schemes can exploit the heterogeneity of the target data rates via users multiplexing and rate allocation. Particularly, users multiplexing of NOMA enables rate splitting onto multiple subcarriers in the overloaded scenario. Moreover, instead of equal rate allocation, our proposed schemes are more flexible to combat the large dynamic range of target data rates via exploiting the frequency diversity. Therefore, for random target data rate, our proposed schemes are more efficient to reduce the power consumption.
%onto different subcarriers according to their target data rates and CNR outage thresholds.
\vspace*{-4mm}
\subsection{Power Consumption versus Number of Users}\label{SimulationD}
\textcolor[rgb]{0.00,0.00,0.00}{Figure \ref{CaseIII_NumUser} illustrates the power consumption versus the number of users with $N_\mathrm{F} = 16$ and ${R}_m^{\mathrm{total}} = 8$ bit/s/Hz, $\forall m$.
The proposed optimal scheme is not included here due to its exponentially computational complexity.
We observe that our proposed scheme is also applicable to underloaded systems with $N_\mathrm{F} > M$, and it is more power-efficient than that of the OMA scheme in both overloaded and underloaded systems.
Furthermore, it can be seen that the power consumption increases with the number of users for all the considered schemes.
This is because a higher power consumption is required when there are more users requiring stringent QoSs.
Besides, our proposed scheme is the most power-efficient among all the schemes.
In particular, compared to the proposed suboptimal scheme, baseline scheme 2 requires a substantially higher power consumption since NOMA requires a careful design of user scheduling to cope with the inherent interference.
On the contrary, baseline scheme 3 needs a slightly higher power than the proposed suboptimal scheme.
As mentioned before, baseline scheme 3 can exploit most of the performance gain of NOMA via enabling multiuser multiplexing with equal rate allocation.}

\begin{figure}[t]
\vspace*{-7mm}
\centering
\includegraphics[width=3.5in]{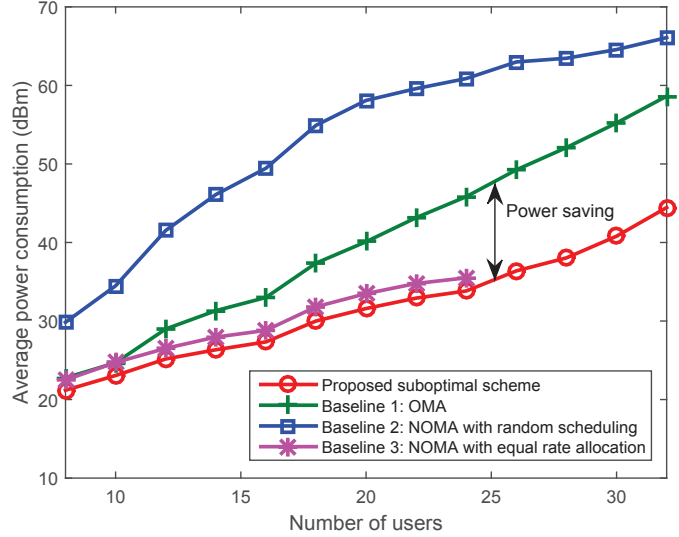}\vspace*{-8mm}
\caption{Average power consumption (dBm) versus number of users with $N_\mathrm{F} = 16$ and ${R}_m^{\mathrm{total}} = 8$ bit/s/Hz. The power saving gain achieved by the proposed scheme over baseline scheme 1 is denoted by the double-side arrow.}
\label{CaseIII_NumUser}\vspace*{-9mm}
\end{figure}

\textcolor[rgb]{0.00,0.00,0.00}{Compared to baseline scheme 1, we can observe that the power saving brought by our proposed suboptimal scheme increases with the number of users.
This can be attributed to the spectral efficiency gain \cite{Ding2015b} and multiuser diversity gain \cite{Sun2016Fullduplex} of NOMA.
On the one hand, NOMA allows multiuser multiplexing on each subcarrier, which provides higher spectral efficiency than that of OMA.
As a result, a smaller amount of power is able to support the NOMA users' QoS requirements than the users using OMA.
Besides, the proposed scheme can efficiently exploit the spectral efficiency gain to reduce the power consumption compared to baseline scheme 1.
In particular, with an increasing number of users, the spectrum available to each user in baseline scheme 1 becomes less due to the exclusive subcarrier allocation constraint, while relatively more spectrum is available in the proposed MC-NOMA scheme owing to the power domain multiplexing.
Consequently, the proposed scheme can save more power compared to the baseline scheme 1.
On the other hand, NOMA possesses a higher capability in exploiting the multiuser diversity than that of OMA.
Particularly, instead of scheduling a single user on each subcarrier in OMA, NOMA enables multiuser multiplexing on each subcarrier, which promises more degrees of freedom for user selection and power allocation to exploit the multiuser diversity.
Therefore, our proposed NOMA scheme with the suboptimal resource allocation design can effectively utilize the multiuser diversity to reduce the total transmit power.
In fact, in the considered MC-NOMA systems, the multiuser diversity comes from the heterogeneity of CNR outage thresholds.
The CNR outage thresholds become more heterogeneous for an increasing number of users. Thus, the power saving gain brought by the proposed NOMA scheme over the OMA scheme increases with the number of users.}
%Note that it is analogous to the case of NOMA with perfect CSIT\cite{Saito2013}, while our defined CNR outage threshold includes the joint effect of channel conditions as well as the QoS requirements.
%On the other hand, the gap between baseline scheme 3 and our proposed schemes
\vspace*{-4mm}
\subsection{Outage Probability}
%\begin{figure}[t]
%\centering
%\includegraphics[width=4in]{figure/CaseV_OutageProbability.eps}
%\caption{Outage probability of our proposed scheme and naive scheme with $N_\mathrm{F} = 8$ and $M = 12$.}
%\label{CaseV_OutageProbability}
%\end{figure}
\begin{figure}[t]
\centering\vspace*{-8mm}
\subfigure[Outage probability for all the users with $\kappa^2_{i,m} = 0.1$.]
{\label{CaseV_OutageProbability:a} %% label for first subfigure
\includegraphics[width=0.46\textwidth]{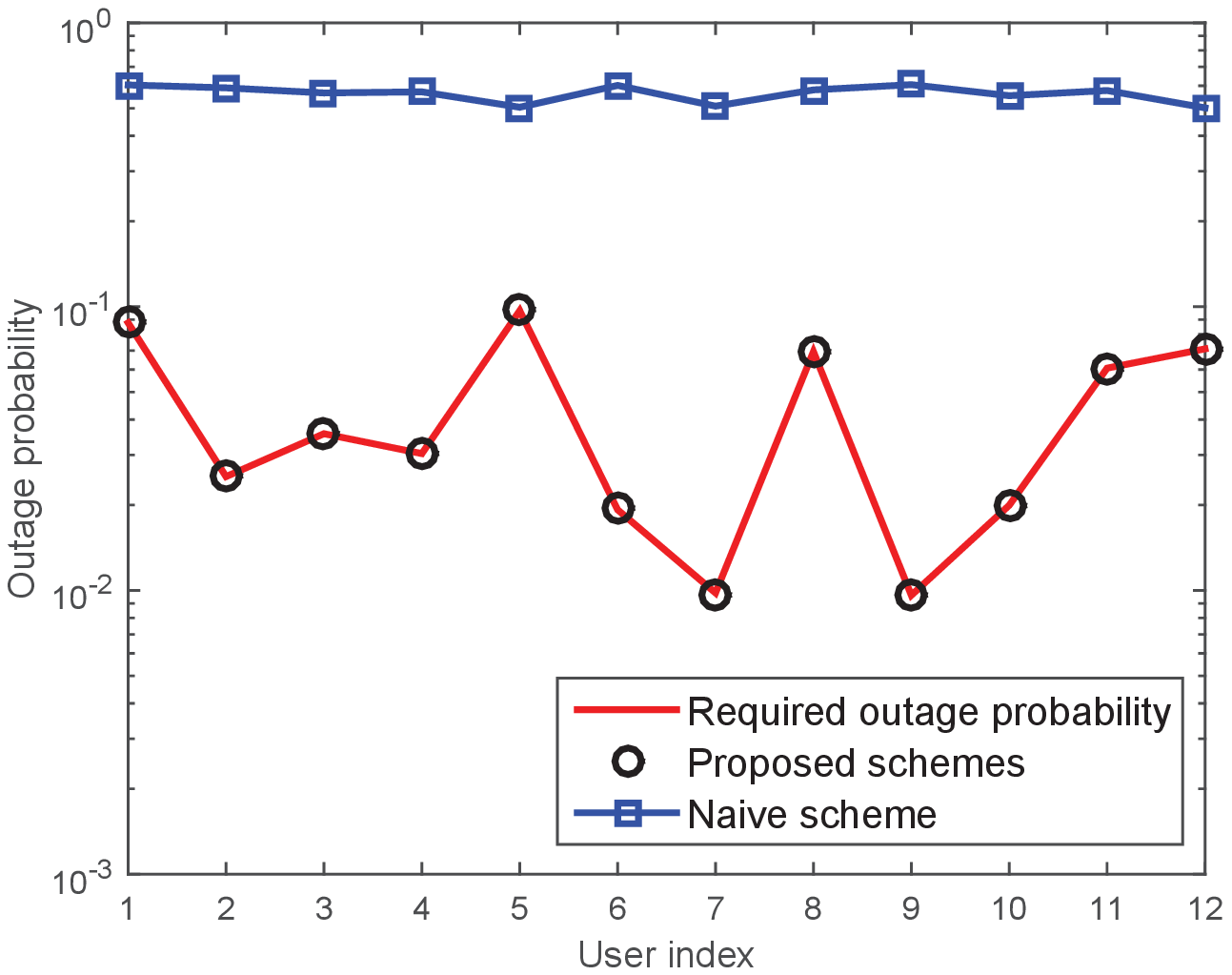}}\vspace*{-2mm}
\subfigure[Outage probability versus $\kappa^2_{i,m}$ for user 9.]
{\label{CaseV_OutageProbability:b} %% label for second subfigure
\includegraphics[width=0.46\textwidth]{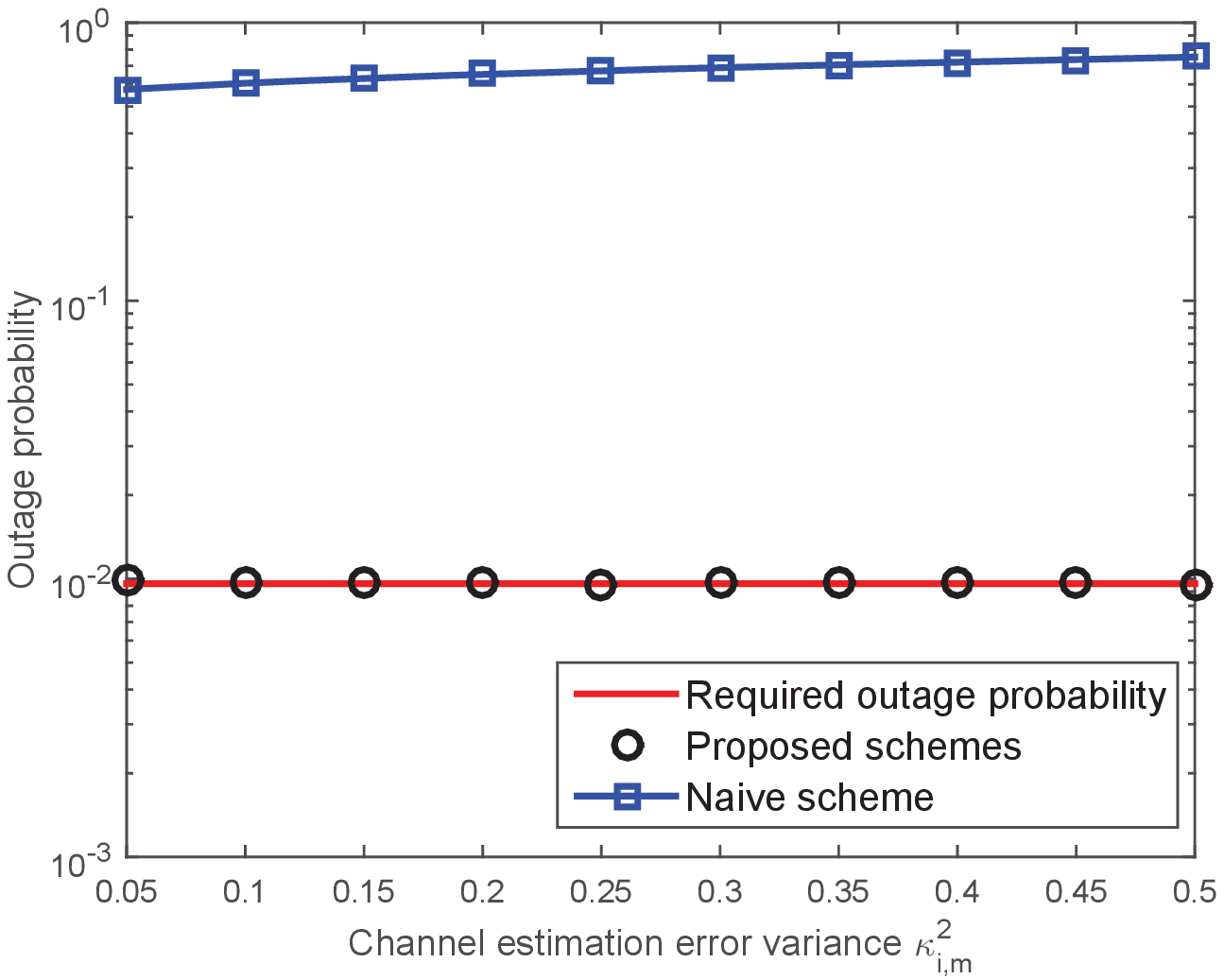}}\vspace*{-2mm}
\caption{Outage probability of our proposed scheme and a naive scheme with $N_\mathrm{F} = 8$ and $M = 12$.}\vspace*{-10mm}
\label{CaseV_OutageProbability}%
\end{figure}

%\begin{figure}[t]
%\footnotesize
%\centering
%\vspace*{-6mm}
%\begin{minipage}{.47\textwidth}
%\centering
%\includegraphics[width=\textwidth]{figure/CaseV_OutageProbability_12.eps}\vspace*{-3mm}
%a) Outage probability for all the users with $\kappa^2_{i,m} = 0.1$.
%\label{CaseV_OutageProbability:a}
%\end{minipage}
%\hspace*{1.5mm}
%\begin{minipage}{.47\textwidth}
%\centering
%\includegraphics[width=\textwidth]{figure/CaseV_OutageProbability_22.eps}\vspace*{-3mm}
%b) Outage probability versus $\kappa^2_{i,m}$ for user 9.
%\label{CaseV_OutageProbability:b}
%\end{minipage}
%\vspace*{-2mm}
%\caption{Outage probability of our proposed schemes and a naive scheme with $N_\mathrm{F} = 8$ and $M = 12$.}\vspace*{-10mm}
%\label{CaseV_OutageProbability}%
%\end{figure}
In this simulation, we introduce a naive scheme where the resource allocation is performed by treating the estimated channel coefficient $\hat{h}_{i,m}$ as perfect CSIT.
Figures \ref{CaseV_OutageProbability:a} and \ref{CaseV_OutageProbability:b} compare the outage probability for our proposed schemes and the naive scheme with $N_\mathrm{F} = 8$ and $M = 12$.
Figure \ref{CaseV_OutageProbability:a} illustrates the outage probability for all the users with channel estimation error variance $\kappa^2_{i,m} = 0.1$.
It can be observed that our proposed schemes can satisfy the required outage probability of all the users while the naive scheme leads to a significantly higher outage probability than the required.
Figure \ref{CaseV_OutageProbability:b} shows the outage probability versus $\kappa^2_{i,m}$ for user 9.
It can be observed that our proposed scheme can always satisfy the required outage probability, despite $\kappa^2_{i,m}$ increases from $0.05$ to $0.5$.
In contrast, the outage probability for the naive scheme increases with $\kappa^2_{i,m}$ due to the deteriorated quality of channel estimates.
In fact, our resource allocation design can guarantee the required outage probability, at the expense of a slightly higher transmit power compared to the case of perfect CSIT, cf. Figure \ref{CaseIV_ChannelError} for $\kappa^2_{i,m} = 0$.

%This is due to enhanced robustness for channel estimation error provided by our particular resource allocation design, which includes the affect of channel uncertainty.
%On the other hand, with larger channel estimation error variance, more instantaneous channel realizations may have larger channel gain than estimation, which results in slightly lower outage probability for the naive scheme with larger $\kappa^2$
\vspace*{-2mm}
\section{Conclusion}
In this paper, we studied the power-efficient resource allocation algorithm design for MC-NOMA systems. The resource allocation algorithm design was formulated as a non-convex optimization problem and it took into account the imperfect CSIT and heterogenous QoS requirements. We proposed an optimal resource allocation algorithm, in which the optimal SIC decoding policy was determined by the CNR outage threshold. Furthermore, a suboptimal resource allocation scheme was proposed based on D.C. programming, which can converge to a close-to-optimal solution rapidly. Simulation results showed that our proposed resource allocation schemes provide significant transmit power savings and enhanced robustness against channel uncertainty via exploiting the heterogeneity of channel conditions and QoS requirements of users in MC-NOMA systems.
\vspace*{-4mm}
\section*{Appendix}
\begin{appendices}
\vspace*{-2mm}
\subsection{Proof of Theorem \ref{Theorem1}}\label{AppendixA}
\vspace*{-1mm}
In the following, we prove Theorem \ref{Theorem1} by comparing the total transmit power for four kinds of SIC policies. Given user $m$ and user $n$ multiplexed on subcarrier $i$, there are following four possible cases on SIC decoding order:
\vspace*{-5mm}
\begin{multicols}{2}
\begin{itemize}
  \item Case I: $u_{i,m} = 1$, $u_{i,n} = 0$,
  \vspace*{-2mm}
  \item Case II: $u_{i,m} = 0$, $u_{i,n} = 1$,
  \item Case III: $u_{i,m} = 1$, $u_{i,n} = 1$,
  \vspace*{-2mm}
  \item Case IV: $u_{i,m} = 0$, $u_{i,n} = 0$,
\end{itemize}
\end{multicols}
\vspace*{-5mm}
\noindent where Case I and Case II correspond to selecting only user $m$ or user $n$ to perform SIC, respectively. Case III and Case IV correspond to selecting both users or neither user to perform SIC, respectively.
In Case I, user $m$ is selected to perform SIC and user $n$ directly decodes its own message. According to \eqref{OutageProbabilityWithSIC} and \eqref{OutageProbabilityWithoutSIC}, the outage probabilities of users $m$ and $n$ are given by
\vspace*{-1mm}
\begin{equation}
    \hspace*{-2mm}{\mathrm{P}}_{i,m}^{{\mathrm{out}}} \hspace*{-1mm}=\hspace*{-1mm} {\mathrm{Pr}}\left\{ {\frac{{{{\left| {{h_{i,m}}} \right|}^2}}}{{\sigma _{i,m}^2}} \hspace*{-1mm}<\hspace*{-1mm} \max \hspace*{-1mm}\left( {\frac{{{\gamma _{i,n}}}}{{{p_{i,n}} \hspace*{-1mm}-\hspace*{-1mm} {p_{i,m}}{\gamma _{i,n}}}},\frac{{{\gamma _{i,m}}}}{{{p_{i,m}}}}} \right)} \hspace*{-1mm}\right\}\; \text{and} \;\;
    {\mathrm{P}}_{i,n}^{{\mathrm{out}}}\hspace*{-1mm} =\hspace*{-1mm} {\mathrm{Pr}}\left\{ \frac{{{{\left| {{h_{i,n}}} \right|}^2}}}{{\sigma _{i,n}^2}} \hspace*{-1mm}<\hspace*{-1mm} \frac{{{\gamma _{i,n}}}}{{{p_{i,n}} \hspace*{-1mm}-\hspace*{-1mm} {p_{i,m}}{\gamma _{i,n}}}}  \right\},\label{OutageProbabilityWithoutSIC2}
\vspace*{-1mm}
\end{equation}
respectively. Note that a prerequisite $p_{i,n} - {p_{i,m}}{\gamma _{i,n}} > 0$ should be satisfied, otherwise the SIC will never be successful, i.e., ${\mathrm{P}}_{i,m}^{{\mathrm{out}}}=1$.

Combining the threshold definition in \eqref{OutageThreshold1} and the QoS constraint \textbf{C5} in \eqref{P1}, the feasible solution set spanned by \eqref{OutageProbabilityWithoutSIC2} can be characterized by the following equations:
\vspace*{-2mm}
\begin{equation}
    p_{i,n} - {p_{i,m}}{\gamma _{i,n}} > 0,\;
    \max \left( {\frac{{{\gamma _{i,n}}}}{{{p_{i,n}} - {p_{i,m}}{\gamma _{i,n}}}},\frac{{{\gamma _{i,m}}}}{{{p_{i,m}}}}} \right) \le  {\beta _{i,m}}, \; \text{and}\;\;
    \frac{{{\gamma _{i,n}}}}{{{p_{i,n}} - {p_{i,m}}{\gamma _{i,n}}}} \le {\beta _{i,n}}.\label{FeasibleSolution13}
\vspace*{-1mm}
\end{equation}
Then, recall that ${\beta _{i,m}} \ge {\beta _{i,n}}$, we have
${p_{i,m}} \ge \frac{{{\gamma _{i,m}}}}{{{\beta _{i,m}}}} \;\; \text{and}\;\;
{p_{i,n}} \ge \frac{{{\gamma _{i,n}}}}{{{\beta _{i,n}}}} + \frac{{{\gamma _{i,m}}{\gamma _{i,n}}}}{{{\beta _{i,m}}}}$.
Then, the optimal power allocation for user $m$ and user $n$ on subcarrier $i$ in Case I are given by
\vspace*{-3mm}
\begin{equation}\label{PowerAllocation12}
    {p_{i,m}^{\mathrm{I}}} = \frac{{{\gamma _{i,m}}}}{{{\beta _{i,m}}}} \;\;  \text{and}\;\;
    {p_{i,n}^{\mathrm{I}}} = \frac{{{\gamma _{i,n}}}}{{{\beta _{i,n}}}} + \frac{{{\gamma _{i,m}}{\gamma _{i,n}}}}{{{\beta _{i,m}}}},
\vspace*{-2mm}
\end{equation}
respectively, with the total transmit power
\vspace*{-3mm}
\begin{equation}\label{PowerComsumptionPerSubcarrier1}
    p^{\mathrm{total}}_{\mathrm{I}} = \frac{{{\gamma _{i,m}}}}{{{\beta _{i,m}}}} + \frac{{{\gamma _{i,n}}}}{{{\beta _{i,n}}}} + \frac{{{\gamma _{i,m}}{\gamma _{i,n}}}}{{{\beta _{i,m}}}}.
\vspace*{-2mm}
\end{equation}

Similarly, we can derive the total transmit power for Cases II, III, and IV as follows:
\vspace*{-3mm}
\begin{align}
p^{\mathrm{total}}_{\mathrm{II}}&= \frac{{{\gamma _{i,n}}}}{{{\beta _{i,n}}}} + \frac{{{\gamma _{i,m}}}}{{{\beta _{i,n}}}} + \frac{{{\gamma _{i,m}}{\gamma _{i,n}}}}{{{\beta _{i,n}}}},\label{PowerComsumptionPerSubcarrier2}\\[-1mm]
p^{\mathrm{total}}_{\mathrm{III}}&= \frac{1}{{1 - {\gamma _{i,m}}{\gamma _{i,n}}}}\left( {\frac{{{\gamma _{i,m}}}}{{{\beta _{i,n}}}} + \frac{{{\gamma _{i,m}}{\gamma _{i,n}}}}{{{\beta _{i,m}}}} + \frac{{{\gamma _{i,n}}}}{{{\beta _{i,m}}}} + \frac{{{\gamma _{i,m}}{\gamma _{i,n}}}}{{{\beta _{i,n}}}}} \right), \; \text{and}\label{PowerComsumptionPerSubcarrier3}\\[-1mm]
p^{\mathrm{total}}_{\mathrm{IV}}&= \frac{1}{{1 - {\gamma _{i,m}}{\gamma _{i,n}}}}\left( {\frac{{{\gamma _{i,m}}}}{{{\beta _{i,m}}}} + \frac{{{\gamma _{i,m}}{\gamma _{i,n}}}}{{{\beta _{i,n}}}} + \frac{{{\gamma _{i,n}}}}{{{\beta _{i,n}}}} + \frac{{{\gamma _{i,m}}{\gamma _{i,n}}}}{{{\beta _{i,m}}}}} \right),\label{PowerComsumptionPerSubcarrier4}
\end{align}
\par
\vspace*{-2mm}
\noindent
where in \eqref{PowerComsumptionPerSubcarrier3} and \eqref{PowerComsumptionPerSubcarrier4}, it is required that $0<1-\gamma _{i,m}\gamma _{i,n}<1$ is satisfied, otherwise no feasible power allocation can satisfy the QoS constraint. Besides, two prerequisites for \eqref{PowerComsumptionPerSubcarrier3} are
\vspace*{-1mm}
\begin{equation}
\hspace*{-2mm}p_{i,m} \hspace*{-1mm}=\hspace*{-1mm} \frac{1}{{1 \hspace*{-1mm}-\hspace*{-1mm} {\gamma _{i,m}}{\gamma _{i,n}}}}\left( \frac{{{\gamma _{i,m}}}}{{{\beta _{i,n}}}} \hspace*{-1mm}+\hspace*{-1mm} \frac{{{\gamma _{i,m}}{\gamma _{i,n}}}}{{{\beta _{i,m}}}}\right) \hspace*{-1mm}\ge\hspace*{-1mm} \frac{{{\gamma _{i,m}}}}{{{\beta _{i,m}}}}\;\; \text{and} \;\;
p_{i,n} \hspace*{-1mm}=\hspace*{-1mm} \frac{1}{{1 \hspace*{-1mm}-\hspace*{-1mm} {\gamma _{i,m}}{\gamma _{i,n}}}}\left( \frac{{{\gamma _{i,n}}}}{{{\beta _{i,m}}}} \hspace*{-1mm}+\hspace*{-1mm} \frac{{{\gamma _{i,m}}{\gamma _{i,n}}}}{{{\beta _{i,n}}}}\right) \hspace*{-1mm}\ge\hspace*{-1mm} \frac{{{\gamma _{i,n}}}}{{{\beta _{i,n}}}}, \label{CaseIII1}
\end{equation}
respectively, otherwise there is no feasible solution.
From \eqref{PowerComsumptionPerSubcarrier1}, \eqref{PowerComsumptionPerSubcarrier2}, \eqref{PowerComsumptionPerSubcarrier3}, and \eqref{PowerComsumptionPerSubcarrier4}, we obtain
\vspace*{-3mm}
\begin{equation}
p^{\mathrm{total}}_{\mathrm{II}} \ge p^{\mathrm{total}}_{\mathrm{I}},\;
p^{\mathrm{total}}_{\mathrm{III}} > p^{\mathrm{total}}_{\mathrm{I}},\; \text{and}\;
p^{\mathrm{total}}_{\mathrm{IV}} > p^{\mathrm{total}}_{\mathrm{I}},\label{Conclusion3}
\vspace*{-3mm}
\end{equation}
which means that the SIC decoding order in Case I is optimal for minimizing the total transmit power. Note that the relationship of $p^{\mathrm{total}}_{\mathrm{III}} > p^{\mathrm{total}}_{\mathrm{I}}$ can be easily obtained by lower bounding $p_{i,n}$ by $\frac{{{\gamma _{i,n}}}}{{{\beta _{i,n}}}}$ according to \eqref{CaseIII1}. Interestingly, we have $p^{\mathrm{total}}_{\mathrm{II}} = p^{\mathrm{total}}_{\mathrm{I}}$ for ${\beta _{i,m}} = {\beta _{i,n}}$, which means that it will consume the same total transmit power for Case I and Cases II when both users have the same CNR outage threshold.

\vspace*{-5mm}
\subsection{Proof of Theorem \ref{Theorem2}} \label{AppendixB}
\vspace*{-1mm}
The constraint relaxed problem in \eqref{P3Continuous} is equivalent to the problem in \eqref{P3} if the optimal solution of \eqref{P3Continuous} still satisfies the relaxed constraints in \eqref{P3}.
For notational simplicity, we define the objective function in \eqref{P3} as
$f\left( {\boldsymbol{{\gamma} }} \right) = \sum\limits_{i = 1}^{{N_{\mathrm{F}}}} {\sum\limits_{m = 1}^M {\frac{{{{\gamma}}_{i,m}}}{{{\beta _{i,m}}}}} }  + \sum\limits_{i = 1}^{{N_{\mathrm{F}}}} {\sum\limits_{m = 1}^{M-1} {\sum\limits_{n = m + 1}^M {\frac{{{{\gamma} _{i,m}}{{\gamma} _{i,n}}}}{{\max \left( {{\beta _{i,m}},{\beta _{i,n}}} \right)}}} } }$.
For the optimal solution of \eqref{P3Continuous}, $\left({{\overline{s}_{i,m}^*}},{\gamma}_{i,m}^*\right)$, $i \in \left\{1, \ldots ,{N_{\mathrm{F}}}\right\}$, $m \in \left\{ 1, \ldots ,M \right\}$, and the corresponding optimal value, $f\left( {\boldsymbol{{\gamma} }^*} \right)$, we have the following relationship according to constraint $\overline{\text{\textbf{C8}}}$ in \eqref{P3Continuous}:
\vspace*{-2mm}
\begin{equation}\label{C5Again}
{\overline{s}_{i,m}^*} =
\left\{
\begin{array}{ll}
1 & \text{if}\;{\gamma}_{i,m}^* > 0,\\
\overline{s}_{i,m}^* \in \left[ {0,\;1} \right] & \text{if}\;{\gamma}_{i,m}^* = 0,
\end{array}
\right.
\vspace*{-2mm}
\end{equation}
where ${\boldsymbol{{\gamma} }^*} \in \mathbb{R}^{N_{\mathrm{F}} M \times 1}$ denotes the set of ${\gamma}_{i,m}^*$.
Note that reducing $\overline{s}_{i,m}^*$ to zero where ${\gamma}_{i,m}^* = 0$ will not change the optimal value $f\left( {\boldsymbol{{\gamma} }^*} \right)$ and will not violate constraints $\overline{\text{\textbf{C1}}}$, $\overline{\text{\textbf{C7}}}$, and $\overline{\text{\textbf{C8}}}$ in \eqref{P3Continuous}.
Through the mapping relationship \eqref{OptimalConvert2}, we have
\vspace*{-2mm}
\begin{equation}
{s}_{i,m}^* \in \left\{0,\;1\right\} \subseteq \left[ {0,\;1} \right],\;
\sum\limits_{m = 1}^M {{{s}_{i,m}^*}} \le \sum\limits_{m = 1}^M {{\overline{s}_{i,m}^*}} \le 2, \; \text{and} \;
{{\gamma}}_{i,m}^* = {{s}_{i,m}^*} {\gamma}_{i,m}^*
\vspace*{-2mm}
\end{equation}
which implies that $\left({{{s}_{i,m}^*}},{\gamma}_{i,m}^* \right)$ is also the optimal solution of \eqref{P3Continuous} with the same optimal objective value $f\left( {\boldsymbol{{\gamma} }^*} \right)$. More importantly, $\left({{{s}_{i,m}^*}},{\gamma}_{i,m}^* \right)$ can also satisfy the constraints \textbf{C1}, \textbf{C7}, and \textbf{C8} in \eqref{P3}.
Therefore, the problem in \eqref{P3Continuous} is equivalent to the problem in \eqref{P3}, and the optimal solution of \eqref{P3} can be obtained via the mapping relationship in \eqref{OptimalConvert2} from the optimal solution of \eqref{P3Continuous}.
\vspace*{-2mm}
\end{appendices}

\bibliographystyle{IEEEtran}
\bibliography{NOMA}

\end{document}